\documentclass[10pt]{article}

\usepackage{amsmath,amssymb,amsthm,enumerate,hyperref}
\usepackage{graphicx}
\usepackage{subfig}
\usepackage{latexsym}
\usepackage{authblk}
\usepackage{color}

\numberwithin{equation}{section}  

\theoremstyle{plain} 

\newtheorem{theorem}{Theorem}[section]
\newtheorem{corollary}[theorem]{Corollary}
\newtheorem{lemma}[theorem]{Lemma}
\newtheorem{proposition}[theorem]{Proposition}

\theoremstyle{definition}
\newtheorem{definition}[theorem]{Definition}

\theoremstyle{remark}

\newcommand{\ra}{\rightarrow}

\newcommand{\abs}[1]{\left\lvert #1 \right\rvert}
\newcommand{\ket}[1]{\lvert #1 \rangle}

\newcommand{\wslim}{{\rm w}^*\mbox{-}\lim}
\newcommand{\conv}[1]{{\rm Conv}\left( #1 \right)}

\newcommand{\CC}{\mathbb C}
\newcommand{\ZZ}{\mathbb Z}

\newcommand{\calA}{\mathcal A}
\newcommand{\calB}{\mathcal B}
\newcommand{\calG}{\mathcal G}

\newcommand{\calH}{\mathcal H}
\DeclareMathOperator{\Tr}{Tr}

\begin{document}

\title{The complete set of infinite volume ground states for Kitaev's abelian quantum double models}
\renewcommand\Affilfont{\itshape\small}
  \author[1]{Matthew Cha}
  \author[1,2]{Pieter Naaijkens}
  \author[1]{Bruno Nachtergaele}
  \affil[1]{Department of Mathematics, University of California, Davis, USA}
  \affil[2]{JARA Institute for Quantum Information, RWTH Aachen University, Germany}
  \date{December 20, 2016}

\maketitle

\begin{abstract}
We study the set of infinite volume ground states of Kitaev's quantum double model on $\ZZ^2$ for an arbitrary finite 
abelian group $G$. It is known that these models have a unique frustration-free ground state.
Here we drop the requirement of frustration freeness, and classify the full set of ground states.
We show that the set of ground states  decomposes into $|G|^2$ different charged sectors, corresponding 
to the different types of abelian anyons (also known as superselection sectors).
In particular, all pure ground states are equivalent to ground states that can be interpreted as describing a single excitation.
Our proof proceeds by showing that each ground state can be obtained as the weak$^*$ limit of finite volume ground states of the 
quantum double model with suitable boundary terms. The boundary terms allow for states that represent a pair of excitations, 
with one excitation in the bulk and one pinned to the boundary, to be included in the ground state space.
\end{abstract}

\section{Introduction}

The past two decades have witnessed a rising interest in topologically ordered states, mainly due to the 
realization that their properties could be useful for fault tolerant quantum computation~\cite{Freedman,Kitaev}.  
The quantum double models introduced by Kitaev demonstrated 
the existence of quantum spin models with short-range interactions that have ground states exhibiting 
topological order and anyonic excitation spectrum. A remarkable feature is that the ground space degeneracy 
depends on the genus of the surface on which the model is defined. 

A characteristic of topologically ordered 
states in two dimensions is the appearance of elementary excitations with braid statistics, called anyons~\cite{Wilczek}.  
Perhaps the most well known model for anyons are the quasi-particle excitations of the fractional quantum 
Hall effect~\cite{ArovasSW,MooreR}. Braid statistics have also been studied in the context of local quantum physics~\cite{FredRS,FrohlichG}
and gauge theories~\cite{BaisDW}. The fusion rules and braiding of anyons is encoded
algebraically as a unitary modular tensor category~\cite{BakalovK}.
In particular, the case of the representation theory of the quantum double for a finite group 
has been well studied~\cite{BaisDW,DijkgraafPR,Kitaev,SzlachV}.

Kitaev's models are very special in the sense that 
they have frustration-free ground states and that the interaction terms in the Hamiltonian are all mutually 
commuting. Trivially, these models have a non-vanishing spectral gap above the ground state. The spectral
gap is an important feature in the classification of topologically ordered ground states into a so-called 
topological phase~\cite{ChenGW,Wen}, and has important implications, such as exponential decay 
of correlations~\cite{HastingsK,NachtergaeleS} and entanglement area laws~\cite{BrandaoH,Hastings}.
It was therefore important to show that this gap does not close for sufficiently small uniform perturbations of these 
models~\cite{BravyiHM,MichalakisP}.

Ground states of quantum lattice models is a well-studied subject. Knowing the set of ground states is essential
for understanding the properties of quantum many-body systems at sufficiently low temperatures. 
In the mathematical analysis of certain statistical mechanics phenomena,  such as equlibrium states, phase transitions, 
superselection sectors and phase classification, it is often necessary or convenient to take the infinite volume limit 
(or thermodynamic limit)~\cite{BachmannO,BratteliR,Naaijkens}. General existence 
and decomposition properties of the set of infinite volume ground states have been mastered for some time~\cite{BratteliKR,BratteliR}. 
The problem of finding the complete the set of  ground states for a given model, and proving that it indeed is the complete set, however, has been solved only in a few cases.

In this work, we study quantum double models for abelian groups, in their implementation as quantum spin
Hamiltonians with short-range interactions as defined by Kitaev~\cite{Kitaev}.
The simplest example is the toric code model, which corresponds to the choice $G = \mathbb{Z}_2$.
The abelian quantum double model is particularly interesting because it has all of the characteristic features 
of topologically ordered systems, while at the same time being simple enough to be tackled directly.
The main features of the model are: it is exactly solvable in the sense that the Hamiltonian can be explicitly 
diagonalized; the dimension of the space of ground states of the models defined
on a compact orientable surface is a topological invariant and corresponds to the number of flat $G$-connections 
on the lattice (up to conjugation); there is a spectral gap above the ground state;
the elementary excitations correspond to quasi-particles with braid statistics,
see~\cite{Bachmann} for a rigorous treatment of these results.

The focus of this paper is the set of infinite volume ground states of the model for finite abelian groups.
Although the quantum double models are exactly solvable in finite volume,
much less is known about the thermodynamic limit. 
The first results in this direction are due to Alicki, Fannes and Hordecki~\cite{AlickiFH}.
They showed that in the case $G =\ZZ_2$, also known as the toric code,
there is a unique frustration-free ground state, which coincides with the translation invariant ground state.
This uniqueness property is not general \cite{GottsteinW}, but is related to topological order in the ground state.
The difficulty of solving the full ground state problem can be understood as follows.
If $\delta$ is the derivation generating the dynamics, one has to find \emph{all} states $\omega$ on the quasi-local algebra $\calA$ of observables that satisfy $\omega(A^*\delta(A)) \geq 0$ for all $A$ in the domain of $\delta$.
It is possible to construct ground states as weak$^*$ limits of finite volume ground states,
but even though the boundary goes to infinity in a sense, the resulting state strongly depends on the boundary conditions chosen for the finite volume ground states.

The main result of this paper is a complete classification of the set of infinite volume ground states
for Kitaev's quantum double model for finite abelian groups.
We find that the set of ground states can be decomposed into $\abs{G}^2$ sectors.
There is a one-to-one correspondence between the ground state sectors and
the superselection sectors defined in \cite{FiedlerN}.
In particular, each sector corresponds to a different anyon type.
The strategy of the proof is to reduce the infinite volume calculation to 
a finite volume calculation.
In particular, we find a boundary term for every finite box 
such that the restriction of any infinite volume ground state to the box 
is a ground state of the finite volume Hamiltonian plus the boundary term.
This strategy is motivated by the fact that infinite volume ground states minimize energy in a local region
among all states that are equivalent in the complement of that region~\cite{BratteliKR},
and resembles the classical Dobrushin-Lanford-Ruelle theory of boundary conditions
for the restriction of infinite volume equilibrium states~\cite{FannesW}.

Our analysis is deeply connected to the notions of topological order, superselection sectors
and the fact that the anyon quasi-particles of the quantum double models are time invariant. 
In analogy with the analysis of Doplicher, Haag and Roberts~\cite{Haag1,Haag2} in local quantum physics~\cite{Haag},
it is possible to analyze the different types of anyons in the system and
to obtain the modular tensor category $\operatorname{Rep}(\mathcal{D}(G))$ describing all their properties.
This was done in~\cite{FiedlerN,Naaijkens} for the abelian quantum double models, and we will use some of these results here.

It is often surprisingly difficult to classify the full set of ground states in the thermodynamic limit.
To our knowledge, the complete ground state problem has only been solved for the one-dimensional 
$XY$-model by Araki and Matsui \cite{ArakiM}, for the one-dimensional $XXZ$-models by Matsui \cite{Matsui} 
and Koma and Nachtergaele \cite{KomaN}, and for finite-range spin chains with a unique frustration free 
matrix product ground state by Ogata~\cite{Ogata3}.
To make progress, one typically has to pair the ground state problem with 
model specific notions; in the $XY$-model it was the Jordan-Wigner transformation to fermions
and in the $XXZ$-model and the frustration-free spin chains it was a connection to zero-energy states \cite{FannesNW,GottsteinW}.
As far as we are aware, our result is the first solution to the ground state problem for a quantum model in two dimensions.

This paper is organized as follows:
In Section~\ref{sec:mainresults} we introduce the model and state the main results.
We describe the superselection sectors in Section~\ref{sec:supersel}
and present a new result that equates a boundary projector to a sum of products of local projectors.
A detailed presentation of the main results and the proofs are in Section~\ref{sec:results}.
We conclude with a discussion of the lack of stability for the infinite volume ground state condition, and 
some difficulties that appear in the analysis for the non-abelian case in Section~\ref{sec:con}.

\textbf{Acknowledgements:} PN has received funding from the European Union's Horizon 2020 research and innovation program under the Marie Sklodowska-Curie grant agreement No 657004. BN was supported in part by the National Science Foundation under Grant DMS-1515850.

\section{Main results}\label{sec:mainresults}

In this section, we recall the definition of the quantum double models we study and present a precise statement of our main results.
A more in-depth discussion can be found in the following sections.

Let $G$ be a finite abelian group and consider the bonds (or edges) $\calB$ of the square lattice $\ZZ^2$, i.e.\ the edges between 
nearest neighbors of points (or vertices)  in $\mathbb{Z}^2$.
We give $\calB$ an orientation by having edges either point up or right.
To each edge $e \in \calB$ we associate a $\abs{G}$-dimensional Hilbert space with an orthonormal basis labeled by group elements 
and denoted by $\ket{g}$. Throughout the paper, we use the notation $\bar{g}$ to denote the inverse element $g^{-1}$.
In general, the model can be defined on any oriented metric graph 
and  for non-abelian groups, see~\cite{BombinMD,Kitaev}. Reversing 
the orientation on a given edge corresponds to the unitary transformation that maps $\ket{g}$  to $\ket{\bar{g}}$ in the state space of that
edge. We use the notation $ \Lambda \subset_f \calB$ to indicate that $\Lambda$ is a {\em finite} subset of $\calB$.
The quantum spin system in $\Lambda$ is defined on the Hilbert space $ \calH_\Lambda := \bigotimes_{x\in \Lambda} \calH_x$
and its algebra of observables is $\calA_\Lambda = \calB(\calH_\Lambda)$.
If $ \Lambda_1 \subset \Lambda_2$ there is a natural inclusion
$ i_{\Lambda_2, \Lambda_1}:\calA_{\Lambda_1} \hookrightarrow \calA_{\Lambda_2}$ mapping $ A \mapsto A\otimes I_{\Lambda_2\backslash \Lambda_1}$.
The maps $i_{\Lambda_2, \Lambda_1}$ are isometric morphisms 
so we will often abuse notation by identifying $i_{\Lambda_2,\Lambda_1}(A)$ with simply $A.$
For this net of algebras, we define the local algebra of observables 
and the quasi-local algebra of observables as, respectively,
\begin{equation}
 \calA_{loc} = \bigcup_{\Lambda \subset_f \calB} \calA_\Lambda, \qquad \calA = \overline{\calA_{loc}}^{\| \cdot \|}.
\end{equation}
An observable $A\in \calA$ is said to be supported on $\Lambda$ if $A \in \calA_\Lambda$.
If $\Lambda$ is the smallest set such that $A \in \calA_\Lambda$ then we say 
$\Lambda$ is the support of $A$.

To define the model we specify the local Hamiltonians and the Heisenberg dynamics on $\calA$.
The interaction terms of the local Hamiltonians are non-trivial only on certain subsets of $\calB$, called stars and plaquettes.
We define a \emph{star} $v$ to be a set of four edges sharing a vertex.
Similarly, a \emph{plaquette} $f$ is the set of four edges forming a unit square in the lattice.
Interaction terms are defined for each star and plaquette by
\begin{equation}
A_v = \frac{1}{\abs{G}} \sum_{g\in G} A_v^g, \qquad \text{ and } \qquad  B_f = B_f^e,
\end{equation}
where the terms $A_v^g$ and $B_f^h$ are defined by their action on basis of simple tensors as shown in the following diagram:
\begin{figure}[h]
\centering
\subfloat{\includegraphics[width=0.45\textwidth]{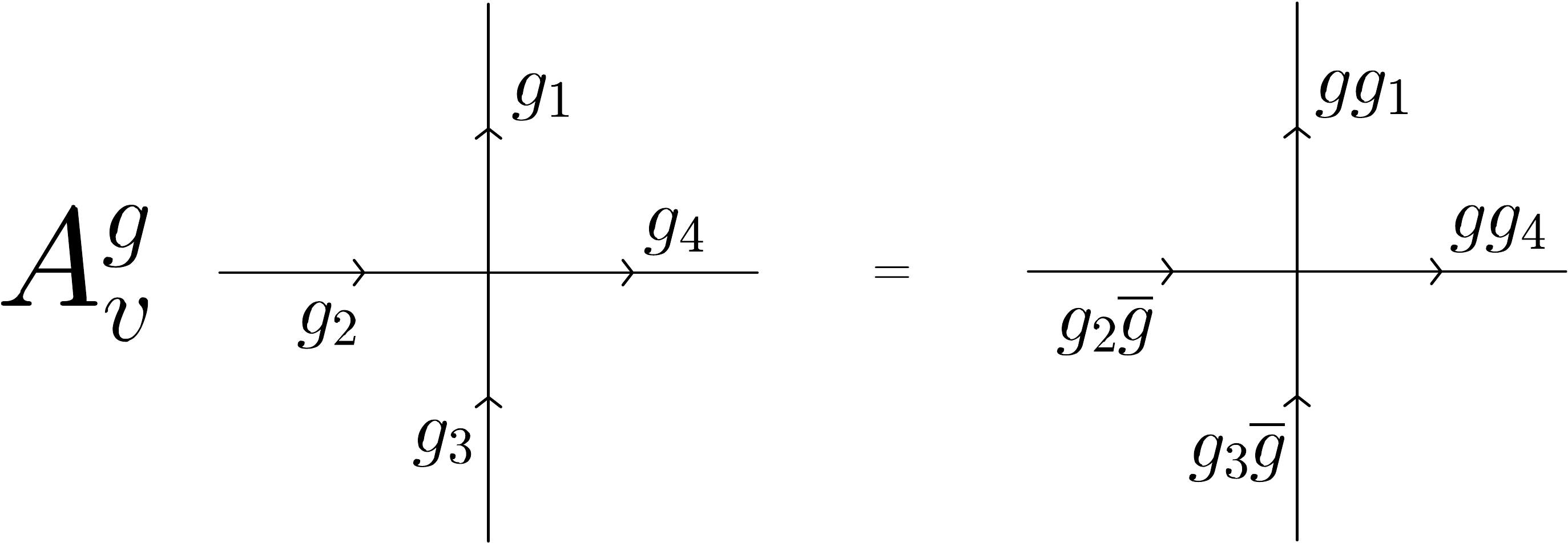}}
\hfill
\subfloat{\includegraphics[width=0.45\textwidth]{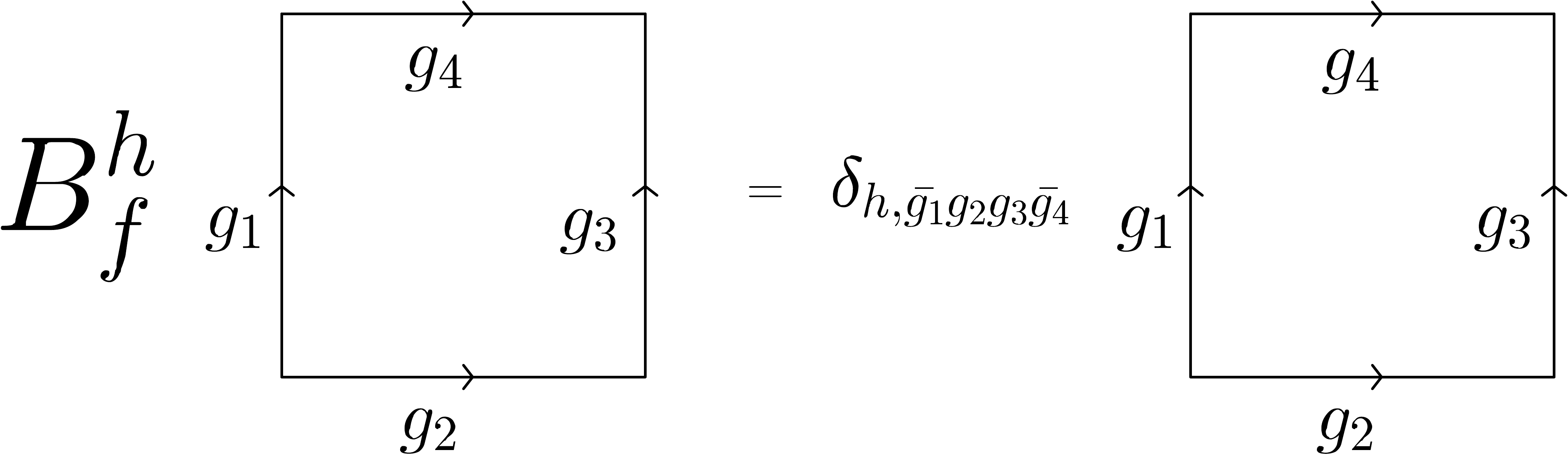}}
\end{figure}\\
Here the group elements $g_1, \dots, g_4$ label the elements of a tensor basis for the local Hilbert space corresponding to a star or plaquette.

It is easy to check that the operators $A_v$ and $B_f$ satisfy the following relations:
\begin{align*}
&A_v^g A_v^{g'} = A_v^{gg'},  & \left(A_v^{\bar{g}}\right)^* = A_v^g, \\
&B_f^h B_f^{h'} = \delta_{h,h'} B_f^h,  &B_f^{h*} = B_f^h, \\
&A_v^g B_f^h = B_f^{gh\bar{g}} A_v^g &\text{(if $v$ and $f$ share edges)}.
\end{align*}
In all other cases the operators commute.

The interactions terms are mutually commuting projectors
\begin{equation}
A_v = A^*_v=A_v^2, \quad \quad B_f = B_f^* = B_f^2, \quad \quad [A_v, B_f] = 0 \text{ for all } v,f.
\end{equation}

We caution the reader that in the case of the toric code model (which corresponds to $G = \mathbb{Z}_2$) one usually shifts the 
local interaction terms by a constant. This has no effect on the dynamics, but the algebraic relations are slightly different. Explicitly,
the common convention is to define the toric code model in terms of star and plaquette operators $A_v^{tc}$ and $B_f^{tc}$ given
by $2 A_v -I = A_v^{tc}$ and $ 2 B_f -I = B_f^{tc}$.

For $\Lambda \subset \calB$ denote the subset of stars and plaquettes contained in $\Lambda$ as
\begin{equation}
\mathcal{V}_\Lambda = \{ v \subset \Lambda: v \text{ is a star} \}, \qquad 
\mathcal{F}_\Lambda = \{ f \subset \Lambda: f \text{ is a plaquette} \}.
\end{equation}
If $\Lambda \subset_f \calB$, the local Hamiltonians for the quantum double models defined by Kitaev \cite{Kitaev} are given by
\begin{equation}
 \sum_{v\in \mathcal{V}_\Lambda} (I - A_v) + \sum_{f\in \mathcal{F}_\Lambda} (I - B_f) = H_\Lambda \in \calA_\Lambda.
\end{equation}
Since the interaction terms are uniformly bounded and of finite range, 
the existence of global dynamics $t \mapsto \tau_t \in \operatorname{Aut}(\calA)$ is readily established.
For our analysis it will be enough to consider squares, $\Lambda_L\subset \calB$, consisting of all edges in $[-L,L]^2$.  
We will denote $H_L = H_{\Lambda_L}$ and $\calH_L = \calH_{\Lambda_L}$.
The generator of the dynamics is the closure of the operator
\begin{equation}
\delta(A) = \lim_{L \ra \infty} [H_L, A],
\end{equation}
where $\calA_{loc}$ is a core for $\delta$, and $\tau_t(A) = e^{i t \delta}(A) $ for all $A \in \calA_{loc}$.

We briefly describe the ground state space, $\calG_L$, of the local Hamiltonians $H_L$.
When defined on a torus, that is when we impose periodic boundary conditions,
the ground states minimize the local energy of each interaction term,
$\calG_L^{per} = \ker H_L^{per} =  \{ \Omega \in \calH_L:  A_s \Omega = \Omega, B_f \Omega = \Omega, \forall s\in \mathcal{S}_L, f\in \mathcal{F}_L;\}$ (see \cite{Kitaev}).
The dimension of $\calG_L^{per}$ is equal to the number of flat $G$-connections up to conjugation and independent of the size of the torus.
For free boundary conditions, the frustration-free property still holds: $\calG_L = \ker H_L$. 
It should be noted though that in this case the dimension grows exponentially with the perimeter of $\Lambda_L$. Other boundary conditions 
have also been considered~\cite{BeigiSW,BravyiK,FreedmanM}.

Recall that a \emph{state} on $\calA$ is a linear functional $\omega: \calA \ra \CC$ such that 
$\omega(A) \geq 0$ if $A\geq 0$  and $ \omega(I) = 1$.  
The set of all states is denoted $ \calA_{+,1}^*$
and is a convex set; its extremal points are called pure states.
By the Banach-Alaoglu theorem the unit ball of all bounded linear functionals on $\calA$ is compact in the weak$^*$ topology.
The positive linear functionals form a convex and weak$^*$ closed subset of this unit ball.
Since $\calA$ is unital, the set of bounded linear functionals $\omega$ for which $\omega(I) = 1$ is closed in the weak$^*$ topology.
Hence the state space is compact in the weak$^*$ topology, being the intersection of a compact with a closed subset (see also~\cite[Thm. 2.3.15]{BratteliR}).

We are interested in the ground states of the infinite system defined as follows.
\begin{definition}\label{defn:gs}
A state $\omega$ is a $\tau$\emph{-ground state} if
\begin{equation}\label{eqn:gs}
\omega( A^* \delta(A)) \geq 0 \quad \quad \text{ for all } A \in \calA_{loc}.
\end{equation} 
We will refer to a $\tau$-ground state as simply a ground state or an \emph{infinite volume ground state}.
\end{definition}

This definition can be interpreted as an infinite volume variational principle expressing
that local perturbations do not decrease the energy of a ground state.
At finite temperature $T$, equilibrium states are defined by the KMS-condition~\cite{HaagHW}.
Definition~\ref{defn:gs} can be obtained as the zero-temperature limit $T\ra 0$ of the KMS-condition.
The set of all grounds states is denoted by $$K = \{ \omega \in \calA_{+,1}^* \mid \forall A\in\calA_{loc}: \omega(A^*\delta(A))\geq 0 \}.$$
$K$ is compact and closed in the weak$^*$ topology, and is a face in $\calA_{+,1}^*$ 
(see Theorem 5.3.37 in \cite{BratteliR}).

Solutions to equation~\eqref{eqn:gs} satisfying in addition a frustration-free condition, i.e.,
$\omega(A_v) = \omega(B_f) = 1$ for all $v \in \mathcal{V}_\calB$ and $ f \in \mathcal{F}_\calB$,
were first studied in \cite{AlickiFH}.
They showed that for $G = \ZZ_2$ there exists a unique frustration-free ground state,
coinciding with the unique translation invariant ground state.
This result was extended for all $G$ in \cite{FiedlerN}.
The notion of topological order is crucial in the proof of these results. 

In~\cite{Naaijkens}, single excitation states were constructed from the frustration-free ground state via localized endomorphisms.
These states are labeled by a pair $(\chi,c) \in \widehat{G} \times G$, where $\widehat{G}$ is the group of characters of $G$.
Each label $(\chi,c)$ denotes a distinct charge, or superselection sector, of the model.
We will show that the single excitation states are solutions to equation~\eqref{eqn:gs}, and in fact that any pure ground state is equivalent to such a state.

To be more precise: in this paper, we prove the set of ground states decomposes into $\abs{G}^2$ sectors, $K^{\chi,c} \subset K$,
corresponding to the superselection sectors defined in \cite{FiedlerN}.
The ground state sectors $K^{\chi, c}$ will be constructed explicitly in Section~\ref{sec:results}.
Heuristically, states in $K^{\chi,c}$ are obtained by projecting ground states onto different charged sectors. 
Our main result is a complete characterization of the set of ground states of the abelian quantum double models,
which is the content of the following theorem:

\begin{theorem}\label{thm:main}
Let $\omega\in K$ be a ground state of the quantum double model.
Then there exists a convex decomposition of $\omega$  as
\begin{equation}
\omega = \sum_{\chi\in \widehat{G},c \in G} \lambda_{\chi,c}(\omega) \omega^{\chi,c} \qquad \text{ where } \quad \omega^{\chi,c}\in K^{\chi,c}.
\end{equation}
For all $(\chi,c)\in \widehat{G}\times G$, $K^{\chi,c}$ is a face in the set of all states.
In particular, if $\omega^{\chi,c} \in K^{\chi,c}$ is an extremal point of $K^{\chi,c}$ 
then $\omega^{\chi,c}$ is a pure state.
If $\omega^{\chi,c} \in K^{\chi,c}$ is a pure state then $\omega^{\chi,c}$ 
is equivalent to a single excitation ground state defined in \cite{FiedlerN},
where two states are equivalent if their corresponding GNS representations are unitarily equivalent.
\end{theorem}

The proof of this result is split up in different theorems, and can be found in Section~\ref{sec:results}.
In particular, the statement and its proof is found as a combined result of Theorem~\ref{thm:qdoubgs1}, Corollary~\ref{cor:face}, and Theorem~\ref{thm:eqstates}.

\section{Excitations and superselection sectors}\label{sec:supersel}
We study the charges, or superselection sectors, of the quantum double model.
To this end, we study the excitations of the model, which can be obtained by using 
what are called \emph{ribbon operators}.
We first recall the definition of ribbon operators and review some properties.
We will just state the properties that will be necessary for the proof of our results, and refer to the existing literature for proofs of these facts~\cite{BombinMD,FiedlerN,Kitaev}.
A good understanding of these operators and how they can be used to build up the local Hilbert spaces is essential to our proof.
In particular, we need to be able to detect such excitations with local projections, which we will call \emph{charge projections}.
We will prove the projections that measure the total charge in a box are supported on the boundary of the box.

After introducing the ribbon operators and charge projectors, we turn to the main topic of interest: the infinite volume ground states. 
In this section we will introduce the unique frustration-free ground state.
We will also construct ground states that are not frustration-free, namely, the single excitation states.
It turns out that they can be obtained by judiciously choosing boundary terms of finite volume Hamiltonians, and taking weak$^*$ limits of finite volume ground states. As such, they can be identified with the superselection sectors of the theory, 
cf.~\cite{FiedlerN,Naaijkens}.
In the next section we will show that in fact all ground states are (quasi-)equivalent to such states.

\subsection{Ribbon operators} 
\begin{figure}
\subfloat{\includegraphics[width=0.4\textwidth]{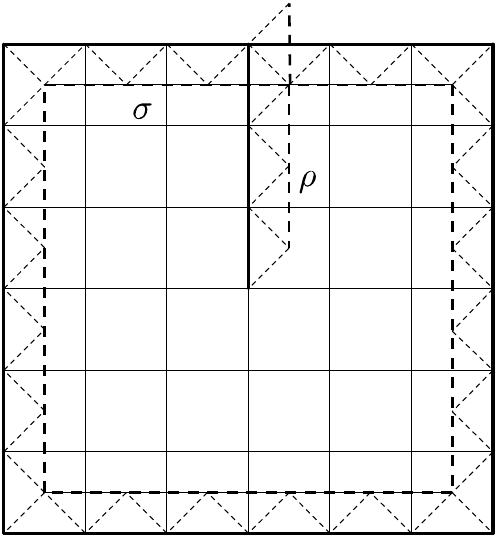}}
\hfill
\subfloat{\includegraphics[width=0.4\textwidth]{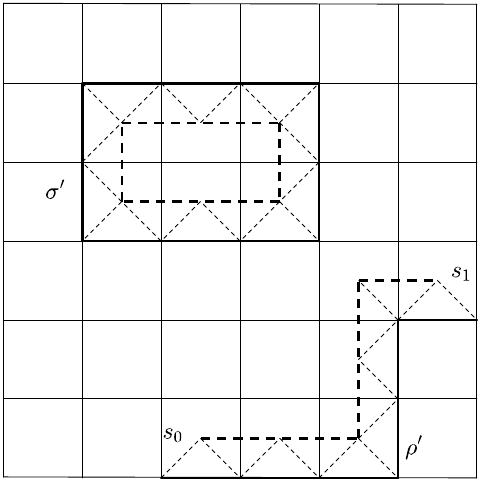}}	
\caption{The region $\Lambda_3$ is depicted with typical configurations of ribbons.
On the left, $\rho$ connects a site in $\mathcal{S}_3$ to a site on the boundary,
and intersects $\sigma = \partial \Lambda_3$, the boundary ribbon of $\Lambda_3$.
On the right, $\rho'$ is an open ribbon connecting site  $s_0 \in \mathcal{S}_3$ to $s_1 \in \mathcal{S}_3$, while $\sigma'$ is a closed ribbon.}
\label{fig:ribbons}
\end{figure}

We abuse notation and use $v$ and $f$ to also denote a vertex and face of the lattice $\ZZ^2$.
A \emph{site} is a pair $s=(v,f)$ of a vertex $v \in \ZZ^2$ and neighboring face $f$. 
Let $\mathcal{S}_L$ denote the set of all sites $s=(v,f)$ such that $v \in \ZZ^2 \cap [-L,L]^2$ and the corresponding face $f \in \mathcal{F}_L$.
We say that a site $s=(v,f)$ is on the boundary of $\Lambda_L$ if $v \in \ZZ^2 \cap [-L,L]^2$ and 
the corresponding face $f \in \mathcal{F}_{L+1} \setminus \mathcal{F}_{L}$.
As we will see, excitations of the model are located at sites.
A \emph{ribbon} $\rho$ is a sequence of adjacent sites connecting two sites $s_0$ and $s_1$. 
We assume ribbons avoid self-crossing 
and label $\partial_0\rho =  s_0$ as the start of the ribbon and $\partial_1 \rho = s_1$ as its end. 
In particular, note that ribbons carry a direction (see~\cite{BombinMD} for how this relates to the direction of the lattice).
We also assume that ribbons have at least two distinct sites.
A ribbon is said to be \emph{open} if $s_0 \neq s_1$ and \emph{closed} if $s_0 = s_1$, see Figure~\ref{fig:ribbons}.

\begin{figure}
	\begin{center}
\includegraphics[width=\textwidth]{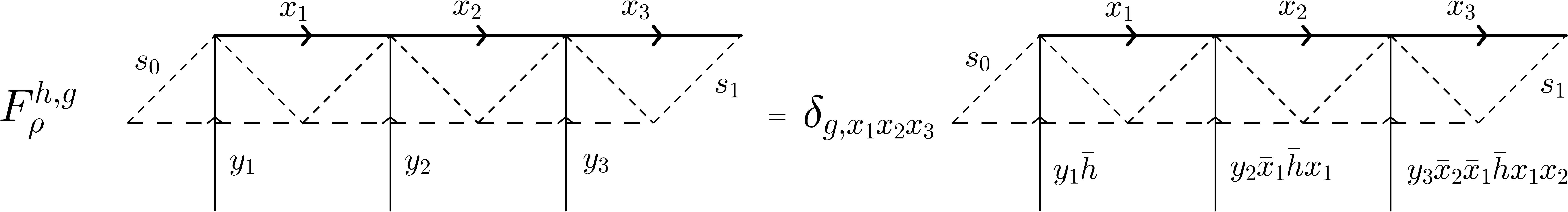}
	\end{center}
\caption{Definition of the ribbon operator $F_\rho^{h,g}$.}
\label{fig:defnrib}
\end{figure}

For any ribbon $\rho$ and $g,h \in G$ the ribbon operator $F_\rho^{g,h}$ is defined as in Figure~\ref{fig:defnrib}.
The ribbon operators can also be defined recursively as concatenations of elementary triangle operators, see \cite{BombinMD}. 
If $\rho_0$ and $\rho_1$ are two ribbons such that $\partial_1 \rho_0 = \partial_0\rho_1$ 
then for the ribbon $\rho = \rho_0\rho_1$,  where the product is defined by concatenation of ribbons,
the ribbon operator satisfies $F_\rho^{h,g}  = \sum_{k\in G} F_{\rho_0}^{h,k} F_{\rho_1}^{\bar{k}  h k, \bar{k} g}$.
As we will see, the operators $F_\rho^{h,g}$ create excitations at the endpoints of the ribbons.
However, in general $F^{h,g}_\rho$ will yield a superposition of different excitation types, and it is more convenient to choose a different basis.
Essentially, what one does is to decompose the space of excitations as invariant subspaces with respect to a \emph{local} action of the quantum double symmetry $\mathcal{D}(G)$ at each site.
This symmetry is implemented by the star and plaquette operators $A^g$ and $B^h$~\cite{BombinMD,Kitaev}.
In this new basis the ribbon operators are labeled by pairs $(\chi, c) \in \widehat{G} \times G$, where $\widehat{G}$ is the group of characters of $G$, and we define
\begin{equation}
F_\rho^{\chi,c} := \sum_{g \in G} \overline{\chi}(g) F_\rho^{\bar{c}, g}.
\end{equation}
If $\rho$ consists of a single edge, then the family of ribbon operators generate the local algebra on that edge.
Similarly, for any finite subset $\Lambda$, the family of ribbon operators supported in $\Lambda$ generate the local observable algebra $\calA_\Lambda$.

We end this overview by listing some of the properties of the ribbon operators that we will use frequently.
Most properties can be verified easily with a direct computation, but see~\cite{BombinMD,FiedlerN,Kitaev} for more information:
\begin{enumerate}[a.)]
\item For operators acting along the same ribbon:
\begin{equation}\label{eqn:ribprop1}
 F^{\chi,c }_\rho F^{\sigma, d}_\rho = F_{\rho}^{\chi\sigma,cd} \qquad \text{and} \qquad (F^{\chi,c}_\rho)^* = F^{\bar{\chi}, \bar{c}}_\rho.
\end{equation}
\item If $\rho$ is an open ribbon connecting sites $\partial_0 \rho = (v_0, f_0)$ and $ \partial_1 \rho= (v_1,f_1)$ then for all $k \in G$ we have
\begin{align}
A_{v_0}^k F_\rho^{\chi,c} &= \chi(k) F_\rho^{\chi,c} A_{v_0}^k, 
& A_{v_1}^k F_\rho^{\chi,c}&= \bar{\chi}(k) F_\rho^{\chi,c} A_{v_1}^k, 
\label{eqn:ribstarrel} \\
B_{f_0}^k F_\rho^{\chi,c} &= F_\rho^{\chi,c} B_{f_0}^{k\bar{c}}, 
& B_{f_1}^k F_\rho^{\chi,c} &= F_\rho^{\chi,c} B_{f_1}^{ck}
.\label{eqn:ribplaqrel}	
\end{align}
In all other cases, the star and plaquette interaction terms commute with the open ribbon operators.
From equations \eqref{eqn:ribstarrel} and \eqref{eqn:ribplaqrel}, we can compute 
the commutation relation with the local Hamiltonian $H_L$ as
\begin{equation}\label{eqn:ribHamrel}
\begin{split}
[H_L, F_{\rho}^{\chi,c}] &= F_{\rho}^{\chi,c} \bigg(B_{f_0} - B_{f_0}^{\bar{c}}  + B_{f_1} - B_{f_1}^c \\
&\quad+ \sum_{k \in G} \left( 1 - \chi(k) \right) A_{v_0}^k + \left(1 - \bar{ \chi}(k)\right) A_{v_1}^k \bigg).
\end{split}
\end{equation}
\item Let $\Omega \in \calG_L$ be a ground state and $\rho$ be an open ribbon.  We can compute the energy introduced by the ribbon operators from the relations \eqref{eqn:ribstarrel}, \eqref{eqn:ribplaqrel}, and \eqref{eqn:ribHamrel},
\begin{equation}\label{eqn:ribenergy}
H_L F_\rho^{\chi,c} \Omega = C_\rho ( 2- \delta_{\chi, \iota} - \delta_{c, e}  ) F_{\rho}^{\chi,c}  \Omega,
\end{equation}
where 
\[ C_\rho = 
\left\{ \begin{array}{ll}
2 & \mbox{ if } \partial_i \rho \in \mathcal{S}_L \text{ for }   i = 0,1 \\
1 & \mbox{ if } \partial_i \rho \in \mathcal{S}_L, \  \partial_{i+1}\rho \notin \mathcal{S}_L \\
0 & \mbox{ if } \partial_i \rho \notin \mathcal{S}_L \text{ for }   i=0,1.
\end{array}\right. \]
Thus, $\rho$ generates \emph{excitations} at its endpoints.
\item If $\rho$ is a closed ribbon then for all $\Omega \in \calG_{L}$,
\begin{equation}\label{eqn:closedribbon}
F_\rho^{\chi,c} \Omega = \Omega \qquad \quad  \forall (\chi, c) \in \widehat{G}\times G
\end{equation}
and 
\begin{equation}\label{eqn:ribclosedcom}
[ F_\rho^{\chi,c}, A_v^g] = [ F_\rho^{\chi,c}, B_f^h] = 0 \qquad \quad \forall v\in \mathcal{V}_\mathcal{B}, f\in \mathcal{F}_\mathcal{B}.
\end{equation}
\item If $\rho = \rho_0\rho_1$, that is, $\partial_1 \rho_0 = \partial_0 \rho_1$ then the ribbon operators obey
\begin{equation}\label{eqn:ribbonconcatenation}
F_\rho^{\chi, c} = F_{\rho_0}^{\chi,c} F_{\rho_1}^{\chi,c}.
\end{equation}
\item A complete set of eigenvectors of $H_L$ for $\calH_L$ is,
\begin{equation}\label{eqn:qdgsspan}
\left\{ \prod_{i} F_{\rho_i}^{\chi_i, c_i} \Omega : \forall \Omega\in \calG_L, (\chi_i, c_i) \in\widehat{G}\times G, \rho_i \text{ a ribbon } \right\}.
\end{equation}
\item  When two ribbons intersect once (as in left Figure~\ref{fig:ribbons}),
\begin{equation}\label{eqn:ribbonrelation}
F_\rho^{\chi,c} F_\sigma^{\xi,d} = \chi(d)\bar{\xi}(c)  F_\sigma^{\xi,d}F_\rho^{\chi,c}.
\end{equation}
In the case of multiple crossings, one can induct on the formula above by decomposing the ribbons $\rho$ and $\sigma$ into sections where only one crossing occurs using the concatenation formula \eqref{eqn:ribbonconcatenation}.
\item  Ribbon operators satisfy \emph{path independence} in the ground state, that is, 
if $\rho$ and $ \sigma $ are ribbons such that $\partial_i \rho = \partial_i\sigma$ for $ i= 0,1$ then
\begin{equation}\label{eqn:ribbonpathind}
F_\rho^{\chi,c}\Omega  = F_\sigma^{\chi,c} \Omega \qquad \text{ for all } \qquad \Omega\in\calG_L.
\end{equation}
\end{enumerate}

\subsection{Local and global charge projectors}\label{sec:finvolume}
We can detect the presence of an excitation and the charge type localized at a site $s=(v,f)$ 
with the orthogonal projectors 
\begin{align}
D_v^\chi &:= \frac{1}{\abs{G}} \sum_{g\in G} \overline{\chi}(g) A_v^g  &\text{ for } \chi \in \widehat{G}, \\
D_f^c &:= B_f^c  &\text{ for } c \in G.
\end{align}
These can be obtained by considering the action of the quantum double at each site~\cite{BombinMD}.
The first detects ``electric'' charges labeled by the characters of $G$, while the latter project on the ``magnetic'' charges labeled by group elements.
The electric charges are located on the vertices, while the magnetic charges are located on faces.
Since we only consider abelian models, these two types of excitations can be treated separately.
One can check that, in the case of abelian groups $G$, the projectors commute.
Thus, there is no ambiguity in defining the operator $D_s^{\chi, c} := D_v^\chi D_f^c$ for the site $s =(v,f)$.

The local charge projectors have the following properties. They follow readily using the properties listed in the previous section:
\begin{align}
D_v^\chi \Omega&= \delta_{\chi,\iota} \Omega, & D_f^c \Omega&=  \delta_{c, e} \Omega &\text{ for all  } \Omega\in \calG_L \label{eqn:localprojgs}\\
D_v^{\chi} F_\rho^{\xi, d} & = F_\rho^{\xi, d} D_v^{\chi\bar{\xi}},  &
D_f^{c} F_\rho^{\xi, d} &= F_\rho^{\xi, d} D_f^{c\bar{d}}
&\text{ if } \partial_0\rho = (v,f) \neq \partial_1\rho  \label{eqn:localprojribbonrelation1}\\
D_v^{\chi} F_\rho^{\xi, d} & = F_\rho^{\xi, d} D_v^{\chi\xi},  &
D_f^{c} F_\rho^{\xi, d} &= F_\rho^{\xi, d} D_f^{dc}
&\text{ if } \partial_1\rho = (v,f) \neq \partial_0\rho\label{eqn:localprojribbonrelation2} \\
D_v^{\chi}D_v^{\xi} &= \delta_{\chi,\xi} D_s^{\chi}, &
D_f^{c}D_f^{d} &= \delta_{c,d} D_f^{c}\\
\sum_{\chi\in\widehat{G}} D_v^{\chi} &= I, &  \sum_{ c \in G} D_f^{c} & = I  &
 \implies \sum_{ \chi\in \widehat{G}, c \in G} D_s^{\chi,c} = I \label{eqn:localprojcomplete}.
\end{align}
Thus the projections onto the charges at a site $s$ form a complete set of orthogonal projections, by equation~\eqref{eqn:localprojcomplete}.
The ground state projector onto the ground state subspace $\calG_L$ is a product of all local charge projectors with trivial charge, $\prod_{s\in \mathcal{S}_L} D_s^{\iota,e}$.

A \emph{global} (or \emph{total}) \emph{charge projector} selects the total charge $(\chi,c)$ in the region $\Lambda_L$.
Heuristically, if for each face $f\in \mathcal{F}_L$ there is a local charge $c_f$, 
then the total charge of magnetic type is $\prod_{f\in \mathcal{F}_L} c_f = c$ (again, we use that $G$ is abelian).
For example, if magnetic charges $c$ and $\bar{c}$ are located on two faces and all other faces carry trivial charge, then the total charge in the region is $ c \bar{c} = e$.
Thus, for a charge $c$ the \emph{conjugate charge} is given by the inverse group element $\bar{c}$.
Here ``conjugate charge'' is standard terminology: it means that you can combine the two charges to obtain a trivial charge. It has nothing to do with conjugation in the group, although here the notions happen to coincide.

Note that the example above in particular shows that trivial global charge does not mean that there are no excitations in the region $\Lambda_L$.
Rather, it means that all charges add up to the trivial charge (similarly, the total electric charge of a state with an electron and a positron is trivial).
An open ribbon operator with both endpoints in $\mathcal{S}_L$  
generates a charge $c$ (see equation~\eqref{eqn:localprojribbonrelation1}), and a conjugate charge $\bar{c}$ (from equation~\eqref{eqn:localprojribbonrelation2}) pair at each of its endpoints. 
Thus, charge is created locally at its endpoints but the total charge of the initial state is preserved.
However, an open ribbon operator with only one of its endpoints in the region $\Lambda_L$ 
does \emph{not} conserve the global charge in the region.
The same is true for charges of electric type, where the multiplication is in the dual group $\widehat{G}$.

We define the global charge projectors by
\begin{equation}
D_L^\chi:= \sum_{\prod_{v}\chi_v = \chi} \ 
\prod_{v\in\mathcal{V}_L} D_v^{\chi_v}, \qquad 
D_L^{c} := \sum_{\prod_{f}c_f= c}  \ 
\prod_{f \in \mathcal{F}_L} D_f^{c_f},
\end{equation}
where the sums are over all configurations $\{ \chi_v \}_{v \in \mathcal{V}_L}$ such that $ \prod_{v}\chi_v = \chi$
and configurations $\{ c_f\}_{f \in \mathcal{F}_L}$ such that $\prod_{f}c_f= c$, respectively.

To project onto (non-trivial) electric $\epsilon$ or magnetic $\mu$ charge types, we have the projectors, respectively,
\begin{align}\label{eqn:globcharge}
D_L^\epsilon  &:= \sum_{\chi\neq \iota}D_L^\chi = I - D_L^\iota, \\
D_L^\mu &:= \sum_{c\neq e} D_L^c = I - D_L^e.
\end{align}
From the definitions above it appears that the global charge projectors are supported on the entire region $\Lambda_L$. 
We will show that they are actually boundary operators.

To this end, we consider closed ribbon operators encircling the boundary of $\Lambda_L$ given by the operators
\begin{equation}\label{eqn:boundribbon}
V_L^\epsilon := \frac{1}{\abs{G}} \sum_{c\in G} \left(I - F_{\partial L}^{\iota, c}\right), \qquad
V_L^\mu := \frac{1}{\abs{G}} \sum_{\chi \in \widehat{G}} \left( I -  F_{\partial L}^{\chi,e} \right),
\end{equation}
where we use $\partial L$ to denote the closed ribbon running along the boundary of $\Lambda_L$, see Figure~\ref{fig:ribbons}.
Physically this can be thought of as an interferometry experiment~\cite{BondersonSS}:
we create a conjugate pair of charges from the ground state, and take one around the region.
Due to the anyonic nature of the charges, if there is a charge present in the region $\Lambda_L$, this is a non-trivial operation for at least one of the charges.
In principle, this can be detected and used to determine the total charge in $\Lambda_L$.

We first show that $V_L^\varepsilon$ and $V_L^\mu$ are in fact projections (cf.~\cite[Sect. B.9]{BombinMD}).
\begin{proposition}
The operators $V_L^\epsilon$ and $V_L^\mu$ are orthogonal projections.
\end{proposition}

\begin{proof}
From equation \eqref{eqn:ribprop1}, we have that 
\[\sum_{\chi,\chi'} F_{ \partial L}^{\chi, e} F_{ \partial L}^{\chi', e}= 
\sum_{\chi,\chi'} F_{ \partial L}^{\chi \chi' , e} = \abs{G} \sum_{\chi} F_{\partial L}^{\chi,e}.\]
Thus,
\begin{align*}
(V_L^\epsilon)^2 &= \frac{1}{\abs{G}^2} \sum_{\chi,\chi'} \left(I - F_{\partial L}^{\chi, e } -  F_{\partial L}^{\chi', e } 
+ F_{\partial L}^{\chi, e } F_{\partial L}^{\chi', e } \right) \\
&= \frac{1}{\abs{G}} \sum_{\chi} \left( I - F_{\partial L}^{\chi, e }  \right) = V_L^\epsilon.
\end{align*}
Also with equation~\eqref{eqn:ribprop1} we find
\begin{align*}
(V_L^\epsilon)^* &= \frac{1}{\abs{G}} \sum_{\chi} I - F_{\partial L}^{\chi, e *}\\
& = \frac{1}{\abs{G}} \sum_{\chi} I - F_{\partial L}^{\bar{\chi}, e} = V_L^\epsilon.
\end{align*}
The calculation to show that $V_L^\mu = (V_L^\mu)^2 = (V_L^\mu)^*$ is similar.
\end{proof}

The following lemma states that localized ribbon operators do not change the total charge.
In other words, local operation cannot change the charged sector.  

\begin{lemma}\label{lem:ribchargeinv}
Let $L' > L$.  Then, $[D_{L'}^\chi,A] = [D_{L'}^c, A] =0 $ for all
 $ \chi\in \widehat{G}, c\in G$ and $A \in \calA_L$.  
\end{lemma}

\begin{proof}
Suppose $A = F_\rho^{\xi,d}$ is a ribbon operator.
If $\rho$ is a closed ribbon, then by equation~\eqref{eqn:ribclosedcom},
$[D_{L'}^\chi, F_\rho^{\xi,d}] = [D_{L'}^c,  F_\rho^{\xi,d}] =0 $.

If $\rho$ is an open ribbon let $\partial_i\rho = s_i = (v_i,f_i)$ for $i = 0,1$. 
Since, $L'>L$ it follows that $v_i \in \mathcal{V}_{L'}$ and $s_i \in \mathcal{S}_{L'}$ for $i =0,1$.
Thus,
\begin{align}
D_{L'}^\chi F_\rho^{\xi,d} 
&= \sum_{\prod_{v}\chi_v = \chi} \ 
\prod_{v\in\mathcal{V}_{L'}} D_v^{\chi_v} F_\rho^{\xi,c}\\
& = F_\rho^{\chi,c}
\sum_{\prod_{v}\chi_v = \chi} \ 
\prod_{\substack{v\in \mathcal{V}_{L'}\\ v\neq v_0,v_1}} D_v^{\chi_v} D_{s_0}^{\chi_{v_0} \bar{\xi} } D_{s_1}^{\chi_{v_1} \xi}\\
& = F_\rho^{\xi, d}  D_{L'}^\chi.
\end{align}
The lemma follows from the fact that ribbon operators generate $\calA_L$.

A similar calculation holds to show $[D_L^c, A] = 0$.
\end{proof}

From the preceding lemma it follows that $D^{\epsilon}_L$ and $D_L^\mu$ are supported on the boundary of $\Lambda_L$.
In fact, it turns out that they are equal to the operators $V^{\epsilon}_L$ and $V_L^{\mu}$ defined in equation~\eqref{eqn:boundribbon}. These two different definitions of what turns out to be the same operators will be convenient for us in the proof of our main result. The following lemma gives a proof that they are indeed equal by showing that they coincide on a spanning set of vectors.

\begin{lemma}\label{lem:globprojboundaryop}
The global charge projectors $D_L^\epsilon$ and $D_L^\mu$, defined above, 
are supported on the boundary of the region $\Lambda_L$.  
In particular, they are equal to the boundary ribbon operators on $\calH_L$, defined in equations~\eqref{eqn:boundribbon}:
\begin{equation}\label{eqn:globalchargewilsonloops}
 D_L^\epsilon= V_L^\epsilon \in \calA_{\Lambda_L\setminus\Lambda_{L-1}}, \qquad  D_L^\mu=V_L^\mu \in \calA_{\Lambda_L\setminus\Lambda_{L-1}}.
\end{equation}
\end{lemma}
\begin{proof}
	The strategy of the proof is to show that $D_L^\epsilon$ and $V_L^\epsilon$ coincide on vectors of the form $\left( \prod F_{\rho_i}^{\chi_i, c_i} \right) \Omega$ for $\Omega \in \calG_L$, by commuting the ribbon operators through $D_L^{\epsilon}$ and $V_L^\epsilon$.\footnote{The result would follow if we can show that the commutation relations of both operators with all ribbon operators coincide, but this is not obvious from a direct computation.} 
By the previous lemma, and the definition of $V^{\epsilon}_L$ and $V^\mu_L$, it follows that all ribbon operators commute with $F_{\rho}^{\chi,c}$ if both of their endpoints are contained in the interior of $\Lambda_L$.
Since we are interested in their action on vectors of the form $\left( \prod F_{\rho_i}^{\chi_i, c_i} \right) \Omega$, by path independence (equation~\eqref{eqn:ribbonpathind}), we can always assume that in that case the ribbon does not intersect the boundary ribbon.
Hence without generality we may assume that each $F_{\rho_i}^{\chi_i, c_i}$ has at least one endpoint on the boundary.

Consider first the case of the empty product.
That is, choose $\Omega \in \calG_L$.  
From equation \eqref{eqn:closedribbon}, $F_{\partial L}^{\chi,c} \Omega = \Omega$ for all $\chi, c$.
Thus, $V_L^\epsilon \Omega = V^\mu_L \Omega = 0$. 
On the other hand, from eqn.~\eqref{eqn:localprojgs}, 
$D_L^\epsilon \Omega= \sum_{\prod_{s}\chi_s \neq \iota} 
\prod_{v} D_v^{\chi_v} \Omega 
= \sum_{\prod_{v}\chi_v \neq \iota}  \prod_{v} \delta_{\chi_v, \iota} \Omega = 0$,
where each term in the sum vanishes because the requirement $\prod_{v}\chi_v \neq \iota$ 
forces a non-trivial charge to exist at at least one star.
Similarly, $D_L^\mu \Omega = 0$.
Hence the operators agree on the subspace $\calG_L$.


Now suppose $\rho$ has both endpoints contained on the boundary.
Since the endpoints of $\rho$ are not in $\mathcal{S}_L$,
the vertex and face projectors corresponding to the endpoints of $\rho$ will not be involved in the products defining $D_L^\epsilon$.
It follows from \eqref{eqn:ribstarrel}
that $D_L^\epsilon F_\rho^{\chi,c} = F_\rho^{\chi,c} D_L^\epsilon$.
On the other hand, from the ribbon crossing relations \eqref{eqn:ribbonrelation} we can compute
\begin{align}
V_L^\epsilon F_{\rho}^{\chi,c} &= \bigg(\frac{1}{\abs{G}}\sum_{g\in G} I - F_{\partial L}^{\iota,g}\bigg)F_{\rho}^{\chi,c}\\
& = F_{\rho}^{\chi,c} \bigg(\frac{1}{\abs{G}}\sum_{\chi\neq\iota} I - \chi(g) \chi(\bar{g}) F_{\partial L}^{\iota,g}\bigg)\\
&= F_{\rho}^{\chi,c} V_L^\epsilon.
\end{align}

Finally, we consider the action of $V_L^\epsilon$ and $D_L^\epsilon$ 
on the spanning set of vectors of the form $\prod_{i}F_{\rho_i}^{\chi_i,c_i} \Omega$.
From the above arguments, without loss of generality, 
we can consider the product of ribbon operators for ribbons which have one endpoint on the boundary
and one in the interior of $\Lambda$, since the other path operators commute with both $V_L^\epsilon$ and $D_L^{\epsilon}$. 
We also assume $\partial_0 \rho_i \in \mathcal{S}_L$.
This can always be achieved by reversing the direction of the ribbon (called ``inversion'' in~\cite{BombinMD}).
By concatenation of ribbons \eqref{eqn:ribbonconcatenation}, without loss of generality we can assume
that all endpoints are distinct, that is,
$\partial_i\rho_k \neq \partial_j\rho_l$ for all $i,j = 0,1$ and $k, l = 0,1,\ldots,n$.
Let $\{(v_i,f_i)\}_{i=1}^n$ enumerate the endpoints of $ \rho_i$ contained in the interior of $\Lambda_L$.
 
Applying the ribbon operator relations \eqref{eqn:ribbonrelation},
\begin{align}
V_L^\epsilon \bigg( \prod_{i=1}^n F_{\rho_i}^{\chi_i,c_i}\bigg) \Omega
&= \frac{1}{\abs{G}} \bigg( \sum_{g\in G} I - F_{\partial L}^{\iota,g}\bigg) 
\bigg( \prod_{i=1}^n F_{\rho_i}^{\chi_i,c_i}\bigg) \Omega\\
& = \frac{1}{\abs{G}}\bigg( \prod_{i=1}^n F_{\rho_i}^{\chi_i,c_i}\bigg)
\bigg( \sum_{g\in G} I - \prod_{i=1}^n \chi_i(g) F_{\partial L}^{\iota,g}\bigg)\Omega\\
& =\bigg( \prod_{i=1}^n F_{\rho_i}^{\chi_i,c_i}\bigg) \left( I - \delta_{\prod_{i=1}^n \chi_i,\iota} \right)\Omega,
\end{align}
where we use orthogonality of characters, $ \frac{1}{|G|} \sum_{g\in G} \prod_{i=1}^n \chi_i(g) = \delta_{\prod_{i=1}^n \chi_i,\iota}$.
On the other hand we have that 
\begin{align}
D_L^\epsilon \bigg( \prod_{i=1}^n F_{\rho_i}^{\chi_i,c_i}\bigg) \Omega
&= \bigg( \sum_{\prod_{v}\chi_v \neq \iota} 
\prod_{v} D_v^{\chi_v} \bigg) \bigg( \prod_{i=1}^n F_{\rho_i}^{\chi_i,c_i}\bigg) \Omega\\
&=\bigg( \prod_{i=1}^n F_{\rho_i}^{\chi_i,c_i}\bigg)  \bigg( \sum_{\prod_{v}\chi_v \neq \iota} 
\prod_{i=1}^n D_{v_i}^{\chi_{i}\chi_{v_i}} \prod_{v \neq v_i} D_v^{\chi_v}\bigg) \Omega\\
&= \bigg( \prod_{i=1}^n F_{\rho_i}^{\chi_i,c_i}\bigg) \left( I - \delta_{\prod_{i=1}^n \chi_i,\iota} \right)\Omega,
\end{align}
where for the last equality we apply an extension of \eqref{eqn:localprojribbonrelation1}.

We have shown $D_L^\epsilon \left( \prod_{i=1}^n F_{\rho_i}^{\chi_i,c_i}\right) \Omega =V_L^\epsilon \left( \prod_{i=1}^n F_{\rho_i}^{\chi_i,c_i}\right) \Omega$ for any arbitrary spanning vector of $\calH_L$.
Therefore, $D_L^\epsilon  = V_L^\epsilon$ as operators on $\calH_L$.
A similar argument gives  $D_L^\mu = V_L^\mu $ as operators on $\calH_L$.
\end{proof}

Before we continue let us briefly summarize the physical interpretation of the properties of ribbon operators and charge projectors.
The ribbon operators create a pair of excitations, one at each end of the ribbon.
In addition, the different types of excitations are labeled by pairs $(\chi,c) \in \widehat{G} \times G$ and the charges at each end are \emph{conjugate} to each other.
In other words, the \emph{global} charge does not change after applying a ribbon operator.
Moreover, when acting with a ribbon operator on the ground state, the resulting state only depends on the endpoints of the ribbon, i.e.\ the location of the excitations.
A similar thing is true for configurations of multiple charges, up to a phase.
This phase can be explained by the anyonic nature of the charges: exchanging two of them gives an overall phase, much like interchanging two fermions yield a minus sign.
The local Hilbert spaces can be obtained completely by such operations, so that we can define a basis by specifying the charge at each site (with an additional constraint on the total charge).
This observation will play an important role in our proof.

\subsection{Infinite volume ground states}
Let $\Omega_L\in \calG_L$ be a sequence of finite volume ground states of $H_L$.
A frustration-free ground state of the quantum double model $\omega^0$
can be constructed as follows.
Consider a family of states $\{\omega_L\}_{L=2}^\infty$ as $L\ra \infty$, 
where $ \omega_L$ is an arbitrary extension
of the state $\langle \Omega_L , \cdot \ \Omega_L  \rangle$ to the quasi-local algebra $\calA$ (in particular, we could choose a product state).
By compactness of $\calA_{+,1}^*$, there exists a convergent subsequence and denote its limit as $\omega^0$.
For any $v$ and $f$, choose $L$ large enough such that $v, f \subset \Lambda_{L_l}$.
Since $\omega_L$ is a ground state for the finite model it follows that $\omega^0(I - A_v) = \omega_L(I - A_v) = 0$ and $\omega^0(I - B_f) = \omega_L(I - B_f) = 0$.
This is the frustration-freeness property.

Ground states of the quantum double model are locally indistinguishable, so weak$^*$ limits of sequences of such states will always converge to the same state.
It turns out that this is the only frustration-free ground state (that is, satisfying equation~\eqref{eqn:gs}) of the model.
We list some properties of this state in the following proposition:

\begin{proposition}(\cite{AlickiFH,Naaijkens,FiedlerN})
Let $\omega^0$ be the frustration-free ground state of the quantum double model obtained as above.
Then,
\begin{enumerate}[(i)]
\item if $\omega$ is a frustration-free ground state then $\omega = \omega^0$,
\item $\omega^0$ is a pure state,
\item Let $(\pi_0, \Omega_0, \calH_0)$ be a GNS-representation for $\omega^0$ and $H_0$ be the GNS Hamiltonian. That is, $H_0$ is the unique self-adjoint (and here, unbounded) operator satisfying 
	$\pi_0 (\tau_t (A)) = e^{i t H_0} \pi_0(A) e^{-it H_0}$, $ H_0 \geq 0,$ and $H_0 \Omega_0 = 0$ (see~\cite[Prop. 5.3.19]{BratteliR}). 
	Then, $\operatorname{spec}(H_0) = 2 \ZZ^{\geq 0}$ with a simple ground state eigenvector $\Omega_0$.
\end{enumerate}
\end{proposition}

The first property is in fact a general property for frustration-free models satisfying a local indistinguishability condition (for the definition see \cite{BravyiHM}).
The frustration-free ground state of the quantum double model satisfies the local indistinguishability condition \cite{FiedlerN}.
The last property follows from an application of strong resolvent convergence (see for example Theorem VIII.24 of \cite{ReedSimon}) and shows that the quantum double model has a spectral gap in its frustration-free ground state.
As a consequence of local indistinguishablity, the gap is stable under local perturbations \cite{BravyiHM,MichalakisP}.

\subsubsection{Single excitation ground states}
In the finite volume, elementary excitations are constructed by violating one of the frustration-freeness ground state conditions.
These excitations must come in pairs since they are generated by open ribbon operators.
By introducing a boundary condition to the Hamiltonian,
we generate ground states which have one excitation in the bulk and one on the boundary.
In the infinite volume, this is equivalent to moving one of the excitations off to infinity, 
thereby isolating a single excitation in the bulk.
By construction, these states will be ground states in the infinite volume.
One way to understand this intuitively is that even though there is an excitation, we cannot lower the energy of the system with local operations.
It is possible to move the excitation around (thereby locally decreasing the energy), but we cannot get rid of it completely with local operations, and the moved excitation will increase the local energy density at its new location.

The idea is to use the projections that were introduced in the previous section to define boundary conditions, which can compensate for the existence of an excitation in the bulk.
Recall that these projections are supported on the boundary of $\Lambda_L$.

\begin{definition}
	\label{def:gsboundary}
Define the following Hamiltonians with boundary condition
\begin{align}
H_L^{\epsilon} &:= H_L - V_L^\epsilon,\\
H_L^{\mu} &:= H_L - V_L^\mu,\\
H_L^{\epsilon, \mu} & := H_L - V_L^\epsilon - V_L^\mu. \label{eqn:boundaryH}
\end{align}
We will sometimes use the index $k$ to denote either $\epsilon, \mu,$ or $ (\epsilon,\mu)$,
and set $V^{\epsilon,\mu}_L := V_L^\epsilon + V_L^\mu$.
\end{definition}

Recall that the boundary terms $V_L^k$ are linear combination of closed ribbon operators and thus commute with each interaction term, and hence the Hamiltonian:
\begin{equation}
[V_L^k, B_f] = [ V_L^k, A_v] = [ V_L^k, H_L] = 0 \quad \text{ for all } \quad k, f, v.
\end{equation}
In what follows we will show that $H_L^k \geq 0$, despite it being a difference of positive operators.
From equation \eqref{eqn:closedribbon}, if  $\Omega_L \in \calG_L$ then $H_L^k \Omega_L = 0$ for all $k$.

Now consider an open ribbon $\rho$ connecting a site $ \partial_0\rho = (v,f)\in \mathcal{S}_L$ to a site on the boundary,
for instance $\rho$ as the ribbon in Figure~\ref{fig:ribbons}, and its corresponding ribbon operator, $F_\rho^{\chi,c}\in \calA_L$.
Then, for $(\chi,c)$ we have that
\begin{align*}
H_L F_\rho^{\chi,c} \Omega_L &= F_\rho^{\chi,c} \left(I- \frac{1}{\abs{G}} \sum_{g\in G} \chi(g) A_v^g + I - B_f^{\bar{c}}\right) \Omega_L\\
& = F_\rho^{\chi,c} \left( 2  -  \delta_{\bar{\chi},\iota}- \delta_{\bar{c},e}\right) \Omega_L.
\end{align*} 
In the last line we used orthonormality: $\langle \chi_1, \chi_2 \rangle := \frac{1}{|G|} \sum_{g \in G} \bar{\chi}_1(g) \chi_2(g) = \delta_{\chi_1, \chi_2}$. Similar calculations, as in the proof of Lemma~\ref{lem:globprojboundaryop}, yield:
\begin{align*}
	V_L^\epsilon F_\rho^{\chi,c} \Omega_L &= F^{\chi,c}_\rho \left[\frac{1}{\abs{G}} \sum_{d\in G} \left(I - \chi(d) F^{\iota, d}_{\partial L}\right)\right] \Omega_L\\
& = F^{\chi,c}_\rho ( I - \delta_{\chi, \iota}) \Omega_L,
\end{align*}
and for the magnetic charges,
\begin{align*}
	V_L^\mu F_\rho^{\chi,c} \Omega_L &= F^{\chi,c}_\rho \left[\frac{1}{\abs{G}} \sum_{\xi\in \widehat{G}}\left( I - \xi(c) F^{\xi,e}_{\partial L}\right)\right] \Omega_L\\
& = F^{\chi,c}_\rho ( I - \delta_{c, e}) \Omega_L.
\end{align*}
Therefore, together with equation~\eqref{eqn:ribenergy}, we find
\begin{equation}
H_L^{\epsilon, \mu} F_\rho^{\chi, c} \Omega_L = 0.
\end{equation}
If $\rho$ connects  two sites in $\mathcal{S}_L$ then $[V_L^k, F_\rho^{\chi,c}] = 0$.
Thus, combined with a similar calculation from above for multiple charges on the boundary, we can conclude that 
\[ \bigg\langle \bigg( \prod_{i} F_{\rho_i}^{\chi_i, c_i} \bigg) \Omega_L,  H_L^{\epsilon, \mu} \bigg( \prod_{i} F_{\rho_i}^{\chi_i, c_i} \bigg) \Omega_L \bigg\rangle \geq 0.\]
Therefore,
\begin{lemma}\label{lem:gsspan}
Let $\calG_L^k$ be the ground state space of $H_L^k$ for $k = \epsilon, \mu$ and $(\epsilon,\mu)$.  Then,
\begin{enumerate}[(i)]
\item $H_L^{\epsilon,\mu} \geq 0$ and $H_L^k \calG_L^k = 0$,
\item $\calG_L^{\epsilon}$ is spanned by $\{ F_\rho^{\chi,e} \Omega_L: \rho $ connects an interior site to the boundary, $\Omega_L \in \calG_L, \chi\in \widehat{G}\}$.

$\calG_L^{\mu}$ is spanned by $\{ F_\sigma^{\iota,c} \Omega_L: \sigma $ connects an interior site to the boundary, $\Omega_L \in \calG_L,  c\in G \}$.

$\calG_L^{\epsilon,\mu}$ is spanned by $\{ F_\rho^{\chi,e}F_\sigma^{\iota,c}  \Omega_L: \rho,\sigma $ connects interior sites to the boundary, $\Omega_L \in \calG_L, (\chi,c)\in\widehat{G}\times G\}$.
\end{enumerate}
\end{lemma}

From the decompositions given in Lemma~\ref{lem:gsspan}, 
it is clear that 
\begin{equation}\label{eqn:gsintersection}
\text{if } L' > L \text{ then } H_L^{ \epsilon, \mu} ( \calG_{L'}^{\epsilon,\mu}) = 0.
\end{equation}
Note that $\calG^\mu_L$ and $\calG_L^\epsilon$ are subspaces of $\calG_L^{\epsilon,\mu}$. 
This result allows us to decompose the ground state space into different sectors corresponding to the different charges:
\begin{corollary}
The ground state space has a natural decomposition
\begin{equation}
\calG_L^{\epsilon,\mu} = \bigoplus_{\chi\in\widehat{G},c\in G} D_L^{\chi,c}\calG_L^{\epsilon,\mu}
\end{equation}
\end{corollary}

\begin{proof}
This follows from Lemma \ref{lem:gsspan} and the relation \eqref{eqn:localprojcomplete}.
\end{proof}

We now come to states in the thermodynamic limit that describe a single excitation.
Such states may be constructed on the quasi-local algebra by moving one of the excitations in a pair off to infinity.
Let $\rho$ be a ribbon extending to infinity such that $\partial_0\rho = s$ and $\partial_1 \rho = \infty$, where $\partial_1 \rho = \infty$ means that the ribbon goes to infinity in any direction.
We assume that it does not ``loop back'', in the sense that if $\rho_n$ is the ribbon consisting of the first $n$ parts of $\rho$, then for any fixed point in the lattice, the distance to the endpoint of $\rho_n$ that is not fixed goes to infinity as $n$ goes to infinity.

We denote $\rho_L = \rho \cap \Lambda_L$.
Define the state $\omega^{\chi,c}_s$  on $\calA_{loc}$, and its unique continuous extension to $\calA$, by
\begin{equation}\label{eqn:singleexcitation}
\omega_s^{\chi,c} (A) := \lim_{L\ra \infty} \langle F_{\rho_L}^{\chi,c } \Omega_L, A   F_{\rho_L}^{\chi,c } \Omega_L\rangle.
\end{equation}

The limit converges because the sequence is eventually constant for fixed local $A$.
That is, by concatenation \eqref{eqn:ribbonconcatenation} and unitarity in the ribbon operators, there exists $L>0$ such that for all $L' >L$ we have 
$(F_{\rho_{L'}}^{\chi,c})^* A F_{\rho_{L'}}^{\chi,c} = (F_{\rho_{L}}^{\chi,c})^* A F_{\rho_{L}}^{\chi,c}$,
and by local indistinguishability, the state is independent of the choice of sequence $\Omega_L$. 
By path independence in the ground state, the state $\omega_s^{\chi,c}$ is also independent of the 
path that $\rho$ takes to infinity and depends only on the basepoint $s$.

Note that by construction we have that $H_L^{\epsilon,\mu} \geq 0$,
$ \delta(A) = \lim_{L\ra \infty} [H_L, A] = \lim_{L\ra \infty} [H_L^{\epsilon,\mu}, A]$  for all $ A \in \calA_{loc}$, 
and  $\omega_s^{\chi,c}( H_L^{\epsilon,\mu}) = 0$ from Lemma~\ref{lem:gsspan}.
From the following basic lemma, it follows that $\omega_s^{\chi,c}$ is an infinite volume ground state. 

\begin{lemma}\label{lem:finitegslim}
Let $\omega \in \calA_{+,1}^*$ and $\tilde{H}_L\in \calA_L$ be a sequence of positive operators 
such that $\delta(A) = \lim_{L \ra \infty} [\tilde{H}_L, A]$ for all $A \in \calA_{loc}$.
If $\omega (\tilde{H}_L) = 0$ for all $L$ then $\omega$ is a ground state, that is, 
$\omega(A^* \delta(A)) \geq 0$ for all $A \in \calA_{loc}$.
\end{lemma}

\begin{proof}
Let $\rho_L \in \calA_L$ be the reduced density matrix for $\omega$ on $\calA_L$, that is, 
$\omega(A) = \Tr(\rho_L A)$ for all $A \in \calA_L$. 
From the condition, $ \omega(\tilde{H}_L) = 0$ and $\tilde{H}_L \geq 0$, it follows that $\tilde{H}_L \rho_L = 0$ for all $L$.
Therefore, by boundedness of $\omega$,
\begin{align*}
\omega(A^* \delta(A)) & = \lim_{L \ra \infty} \omega( A^* [ \tilde{H}_L, A])\\
& = \lim_{L \ra \infty}  \omega(A^* \tilde{H}_L A ) - \omega(A^* A \tilde{H}_L) \\
& = \lim_{L \ra \infty} \Tr(\rho_L A^* \tilde{H}_L A)  - \Tr( \tilde{H}_L \rho_L A^*A)\\
& \geq 0
\end{align*}
for all $A \in \calA_{loc}$.
\end{proof}

The states $\omega^{\chi,c}_s$ were first introduced in~\cite{FiedlerN,Naaijkens}.
They showed the states can be constructed from the frustration-free ground state via an automorphism,
$ \omega_s^{\chi,c} = \omega^0 \circ \alpha^{\chi,c}_\rho$, where
\begin{equation}\label{eqn:chargemorp}
\alpha^{\chi,c}_\rho(A) = \lim_{L\ra \infty} F_{\rho_L}^{\chi,c *} A \ F_{\rho_L}^{\chi,c}.
\end{equation}
The limit converges in norm for each $A \in \calA$ and defines an outer automorphism.

\subsection{Superselection criterion}
Superselection sectors arise because the quasi-local algebra has many inequivalent representations.
Most representations do not have any physical relevance (for example, because the energy is unbounded),
so it is important to restrict the class of representations of interest.
For example, a theory may have different, inequivalent particle types, like the excitations in the quantum double.
Another example would be electric charge.
Here we will use the term ``charge'' in a generalized sense, as a label of the different particle types.
Once we can identify different classes of representations with charges, it is reasonable to impose additional constraints.
In particular, we can impose certain locality conditions, and demand that we are able to move the localization regions around.
A superselection criterion is a rule that tells us precisely which representations we select in the end.

The Doplicher-Haag-Roberts (DHR) analysis in algebraic quantum field theory showed that 
starting from a vacuum state and a physically motivated superselection criterion, 
one could recover a family of superselection sectors corresponding to the global gauge group \cite{Haag1,Haag2}.
This allows one to recover all physically relevant properties of the charges, such as their particle statistics.
A similar analysis has been done for the quantum double models, 
producing the single excitation ground states as the 
irreducible objects in each superselection sector \cite{FiedlerN,Naaijkens}.
The role of the vacuum is played by the translation invariant frustration-free ground state.

For the quantum double models the relevant criterion is as follows.
Let $\Lambda \subset \mathcal{B}$ be an infinite cone region (the precise shape is not that important).
We consider representations $\pi$ which satisfy the following criterion for \emph{any} such $\Lambda$:
\begin{equation}\label{eqn:conecrit}
\pi_0 \upharpoonright \calA_{\Lambda^c} \cong \pi \upharpoonright \calA_{\Lambda^c}.
\end{equation}
Here $\pi_0$ is the GNS representation of the frustration-free ground state and $\pi_0 \upharpoonright \calA_{\Lambda^c}$ means that we restrict the representation to $\calA_{\Lambda^c}$, the $C^*$-algebra generated by all local observables supported outside $\Lambda$.
Physically, to detect the charge of a state in the representation $\pi$, one needs to measure the value of a ``Wilson loop''.
If such loops around the charge are not allowed (as in the selection criterion, due to the absence of the cone), the charge cannot be detected.

The superselection structure of the quantum double model can be analyzed in the same spirit as the DHR program.
The sector structure is summarized in the following proposition

\begin{proposition}\label{prop:singleexc}(\cite{FiedlerN,Naaijkens})
Let $ (\pi_s^{\chi,c}, \Omega_s^{\chi,c}, \calH_s^{\chi,c})$ be the GNS triple for $\omega_s^{\chi,c}$.
Then,
\begin{enumerate}[(i)]
\item $\pi_s^{\chi,c}$ are irreducible representations satisfying the criterion~\eqref{eqn:conecrit},
\item $\pi_{s}^{\chi,c} \cong \pi_{s'}^{\chi,c} $,
\item if $(\chi,c) \neq (\chi',c')$ then $\pi_s^{\chi,c} $ and $\pi_s^{\chi',c'}$ belong to different superselection sectors (and hence are inequivalent),
\item if $\pi$ is irreducible and satisfies \eqref{eqn:conecrit} then there exists $\rho$ and $(\chi, c)$ such that $\pi \cong \pi_s^{\chi,c}$.
\end{enumerate}
\end{proposition}
Pushing this analysis further, all properties of the charges such as their fusion and braiding rules can be recovered~\cite{FiedlerN}.
It follows that the structure is completely described by the representation theory of the quantum double, $\operatorname{Rep}(\mathcal{D}(G))$.
It is interesting to see that the charge superselection structure is closely related to the classification of ground states of the quantum double, as will become even clearer in the next section.

\section{The complete set of ground states}\label{sec:results}
In this section we prove our main result: a complete classification of the ground states of the quantum double model
for abelian groups.
Our strategy is to find a boundary condition such that any 
infinite volume ground state has zero energy for the Hamiltonian with this boundary condition.
It turns out that this is possible with the boundary conditions introduced in Section~\ref{sec:finvolume}.
The classification of infinite volume ground states then
simplifies to a classification of infinite volume limits of finite volume ground states.
These finite volume ground states are well understood by the results in the previous section, and this allows us to obtain our classification.
This strategy is similar to the solution of the complete ground state for the XXZ chain given in \cite{KomaN}.

We begin by introducing some notation.
\begin{definition}
Let $K := \{ \omega \in \calA_{+,1}^* \mid  \omega(A^* \delta(A))\geq 0\}$ denote the set of infinite volume ground states,
where $\delta$ is the generator of the dynamics for the quantum double model for abelian group $G$.
Similarly, for the set of all finite volume ground state functionals of $H_L^{\epsilon,\mu}$ we write $ K_L := \{\omega_L:\calA_L \ra \CC \mid \omega_L(H_L^{\epsilon,\mu})=0\} $.
\end{definition}

The first step is to show that any infinite volume ground state minimizes the energy of the finite volume Hamiltonians $H_L^{\epsilon,\mu}$ of Definition~\ref{def:gsboundary}.
Here we will use the formulation of the boundary term in terms of a sum of products of local charge projections. This gives us precise control on the location of possible excitations.
\begin{lemma}\label{lem:gshambound}
Let $\omega\in \calA_{+,1}^*$.
Then, $\omega \in K$ if and only if for all $L \geq 2$
\begin{equation} 
\omega(H_L - D_L^\epsilon - D_L^\mu) = 0.
\end{equation} 
\end{lemma}

\begin{proof}
($\impliedby$) This follows from Lemma \ref{lem:finitegslim}.

($\implies$) We will show that $\omega( \sum_{v\in\mathcal{V}_L} \left(I - A_v\right) ) = \omega(V_L^\epsilon)$ and similarly that $\omega(\sum_{f\in\mathcal{F}_L} \left(I - B_f\right))  = \omega( V_L^\mu)$.
The result then follows from Lemma~\ref{lem:globprojboundaryop}.

Let $L\geq 2$ be given.
Consider an arbitrary enumeration of the set of plaquettes, $\mathcal{F}_L = \{f_i\}_{i=1}^{n_L}$, and
a configuration of magnetic charges, $\{ c_i\in G\}_{i=1}^{n_L}$
such that $\prod_i c_i = e$.
In the following, sums and products indexed by $i,j$ and $k$ will run from $1$ to $n_L$ unless otherwise stated.
Pair $f_i$ with a neighboring vertex $v_i$ and 
let $\rho_i $ be a ribbon 
such that $\partial_0 \rho_i = (v_i, f_i) $  and $\partial_1 \rho_i = (v_{i+1},f_{i+1})$.
With this choice, consider the operator $A =\big( \prod_i F_{\rho_i}^{\iota, \tilde{c}_i}\big)\big(\prod_i(I - B_{f_i})\big)$, 
where the family $\{ \tilde{c}_i \}$ is chosen such that 
\begin{equation}
B_{f_i} \bigg(\prod_k F_{\rho_k}^{\iota, \tilde{c}_k}\bigg) = \bigg(\prod_k F_{\rho_k}^{\iota, \tilde{c}_k}\bigg) B_{f_i}^{c_i}, \quad \forall i.
\end{equation}
Indeed, the $\tilde{c}_i$'s must be such that $c_{i+1} = \tilde{c}_{i} \bar{\tilde{c}}_{i+1}$.
The condition that $\prod_i c_i = e$ guarantees that such a family $\{ \widetilde{c}_i \}_i$ always exists,
for instance, $\tilde{c}_i = \prod_{j\leq i}\bar{c}_j.$

We want to apply the ground state condition to the operator $A$, hence we compute
\begin{align*}
A^*\delta(A) &= A^*[H_L, A] \\
&=  \prod_{i} (I - B_{f_i})\bigg(\prod_i F_{\rho_i}^{\iota, \tilde{c}_i}\bigg)^* \Big(\sum_{j} \bigg[  - B_{f_j},\prod_i F_{\rho_i}^{\iota, \tilde{c}_i} \bigg]\Big) \prod_i (I - B_{f_i})\\
& = \prod_{i} (I - B_{f_i})\bigg(\prod_i F_{\rho_i}^{\iota, \tilde{c}_i}\bigg)^*\prod_i F_{\rho_i}^{\iota, \tilde{c}_i} \bigg(\sum_{j} B_{f_j} - B_{f_j}^{c_{j}}\bigg) \prod_{i} (I - B_{f_i})\\
&= \prod_{i} (I - B_{f_i}) \bigg(\sum_{j} B_{f_j} - B_{f_j}^{c_{j}}\bigg) \\
&=  - \prod_{i} (I - B_{f_i}) \bigg(\sum_{j} B_{f_j}^{c_{j}}\bigg).
\end{align*}
The operator $\prod_{i} (I - B_{f_i}) \bigg(\sum_{j} B_{f_j}^{c_{j}}\bigg)$ is a product of
commuting positive operators and, hence, it is positive. But this implies that $A^*[H_L, A] \leq 0$.

Because of the ground state condition, equation \eqref{eqn:gs}, and the calculation above, $\omega(A^*[H_L, A]) = 0$.
We can then sum over each configuration $c_i$ with trivial product.
Note that if we fix $c_j$ for $j=1, \ldots, n_L-1$, this fixes $c_{n_L}$ by the condition that their product should be trivial.
Hence the summation over all configurations gives
\begin{equation}
	0 = \sum_{(c_1, \dots, c_{n_L-1}) \in G^{n_{L}-1}} \omega\bigg(\prod_{i} (I - B_{f_i}) \bigg(B_{f_{n_L}}^{\overline{\prod c_i}} + \sum_{k=1}^{n_L-1} B_{f_k}^{c_{k}}\bigg)\bigg).
\end{equation}
Here we separated the $n_{L}$ face from the others in the summation, since its magnetic charge is fixed by the others.
We now do the summation over $c_1$.
Note that as $c_1$ runs over the group $G$, so does $\overline{\prod_{i=1}^{n_L-1} c_i}$.
Also note that for any $j$, $\sum_{c_j\in G} B_{f_j}^{c_j} = I$.
This yields, by repeating this procedure,
\begin{align}
0 & = \sum_{(c_2, \dots c_{n_L-1}) \in G^{n_{L}-2}} \omega\bigg(\prod_{i} (I - B_{f_i}) \bigg(2 I + \sum_{k=2}^{n_L-1} B_{f_k}^{c_{k}}\bigg)\bigg)\\
& = c(G, n_L) \omega\bigg(\prod_{i} (I - B_{f_i})\bigg),
\end{align}
where $c(G, n_L)$ is some non-zero constant depending only on $|G|$ and the number of plaquettes.
Therefore, 
\begin{equation}\label{eqn:doubgscond3}
\omega\bigg(\prod_{i=1}^{n_L} (I - B_{f_i})\bigg) = 0.
\end{equation}
Equation \eqref{eqn:doubgscond3} generally holds for a finite subset $\Lambda\subset \calB$,
where we assume that the subset is contained in some box $\Lambda_L$.
We will need this fact for the following argument.

We proceed by induction to show that 
\begin{equation}\label{eqn:doubgscond1}
\omega\bigg( \sum_{f\in \mathcal{F}_L} I-B_f\bigg)  = \omega\bigg( I - \prod_{f\in \mathcal{F}_L} B_f\bigg).
\end{equation}
For the case of two faces, $f_1$ and $f_2$, we have from equation~\eqref{eqn:doubgscond3}
\begin{align*}
0 & = \omega( (I- B_{f_1})(I-B_{f_2}))\\
& = \omega( I - B_{f_1} - B_{f_2} + B_{f_1}B_{f_2}),
\end{align*}
so that $\omega( I - B_{f_1}) + \omega(I- B_{f_2}) = \omega(I- B_{f_1}B_{f_2})$.

Suppose that equation~\eqref{eqn:doubgscond1} holds if $X$ is a finite collection of faces with $\abs{X} \leq n$.
Now let $X$ be a finite collection of $n$ faces making up a region in $\Lambda_L$ and 
enumerate the elements, $X = \{f_i\}_{i=1}^n$ and let $ f_{n+1} \notin X$ be a face in $X$ but otherwise arbitrary.
From equation~\eqref{eqn:doubgscond3} it follows that 
\[
\omega\Big( \Big(\prod_{i\in X} I-B_{f_i}\Big)(I-B_{f_{n+1}} )\Big) = 0.
\]
Expanding the product and using the hypothesis we have,
\begin{align*}
0 & =  \omega\Big( \Big(\prod_{f_i\in X} I-B_{f_i}\Big)(I-B_{f_{n+1}} )\Big) \\
& = \omega\Big( I - \sum_{i=1}^{n+1} B_{f_i} + \sum_{i<j \leq n+1} B_{f_i}B_{f_j} - \sum_{i<j<k \leq n+1} B_{f_i}B_{f_j}B_{f_k} +  \\ & \quad\quad\quad\quad\quad+ \ldots + 
(-1)^{n+1} \Big( \prod_{f_i\in X} B_{f_i} \Big) B_{f_{n+1}}\Big)\\
&= 1 + \left[ \sum_{i=1}^{n+1} \omega\left( I - B_{f_i} \right) - \binom{n+1}{1} \right] - 
\left[ \sum_{i < j \leq n+1} \omega\left(I - B_{f_i} B_{f_j}\right) - \binom{n+1}{2}\right] +\\ 
		 & \quad\quad\quad\quad\quad+ \ldots (-1)^{n+1} \omega\Big( \Big( \prod_{f_i\in X} B_{f_i} \Big) B_{f_{n+1}}\Big)\Big) \\
&= -(-1)^{n+1} + \sum_{i=1}^{n+1} \omega\left( I - B_{f_i} \right)  + \sum_{i < j \leq n+1} \omega\left(I - B_{f_i} B_{f_j}\right) + \ldots + \\
& \quad\quad\quad\quad\quad + (-1)^{n+1} \omega\Big( \Big( \prod_{f_i\in X} B_{f_i} \Big) B_{f_{n+1}}\Big)\Big),
\end{align*}
where in the last step we use the elementary equation $\sum_{k=1}^{n-1} (-1)^k \binom{n}{k} = - (1+(-1)^n)$.
We can then apply the induction hypothesis to all but the last terms.
Note that for the term with $k$ products of $B_{f_i}$, after applying the summation in the induction hypothesis,
each term $\omega(I-B_{f_i})$ appears exactly $\binom{n}{k-1}$ times.
Hence we obtain
\begin{align*}
	0 & =  -(-1)^{n+1}  + \sum_{k=1}^n (-1)^{k+1} \binom{n}{k-1} \left( \sum_{i=1}^{n+1} \omega\left(I - B_{f_i}\right) \right) + \\
	& \quad\quad\quad\quad + (-1)^{n+1} \omega\Big( \Big( \prod_{f_i\in X} B_{f_i} \Big) B_{f_{n+1}}\Big)\Big) \\
	& = -(-1)^{n+1}\omega\left(I - \prod_{i=1}^{n+1} B_{f_i} \right) + (-1)^{n+1} \omega\left( \sum_{i=1}^{n+1} \left(I - B_{f_i}\right)\right),
\end{align*}
where we used that $\sum_{k=1}^{n} (-1)^{k+1} \binom{n}{k-1} = (-1)^{n+1}$.
Therefore equation~\eqref{eqn:doubgscond1} holds. 

Now consider a configuration of magnetic charges, $\{ c_i\in G\}_{i=1}^{n_L}$
such that $\prod_i c_i = e$ and $\rho_i$ and $\widetilde{c}_i$ are as defined earlier.
Let  $ A' = \prod_{i=1}^{n_L} F_{\rho_i}^{\iota, \widetilde{c}_i} \prod_{i=1}^{n_L} B_{f_i}^{c_i}$
and let $l = \#\{i: c_i \neq e\}$.
We compute
\begin{align*}
A'^*\delta(A')& = A'^*[H_L, A'] \\
& = \prod_i B_{f_i}^{c_i}\bigg( \prod_i F_{\rho_i}^{\iota, \widetilde{c}_i} \bigg)^*\sum_{j}  \bigg[ - B_{f_j},  \prod_i F_{\rho_i}^{\iota, \widetilde{c}_i}\bigg]\prod_i B_{f_i}^{c_i}\\
&=\prod_i B_{f_i}^{c_i} \sum_{j} \left(B_{f_j} - B_{f_j}^{c_{j}}\right)\\
& = - l \prod_i B_{f_i}^{c_i}\leq 0.
\end{align*}
Therefore, applying the ground state condition gives
\begin{equation}\label{eqn:doubgscond2}
\text{ if } l > 0 \qquad \text{ then } \qquad  \omega\bigg(\prod_{i=1}^{n_L}B_{f_i}^{c_i}\bigg) = 0.
\end{equation}
Finally, applying the equivalence in Lemma~\ref{lem:globprojboundaryop}
with equations~\eqref{eqn:doubgscond1} and~\eqref{eqn:doubgscond2} gives the result,
\begin{align*}
\omega\bigg( \sum_{i=1}^{n_L} I-B_{f_i}\bigg) & = \omega\bigg( I - \prod_{i=1}^{n_L} B_{f_i}\bigg)\\
& = \omega\bigg( I - \sum_{\prod_i c_i =e }  \prod_{i=1}^{n_L} B_{f_i}^{c_i} \bigg)\\
& = \omega(D_L^\mu) = \omega(V_L^\mu).
\end{align*}

A similar argument gives
\begin{equation}
\omega \bigg( \sum_{v\in \mathcal{V}_L} I  - A_v \bigg) = \omega( V_L^\epsilon).
\end{equation}
This concludes the proof.
\end{proof}

We now state and prove the main result of the paper,
starting with the definitions of the infinite volume ground state subsets:
\begin{definition}\label{def:chargedgs}
Define the following convex subset of states for each $(\chi,c) \in \widehat{G} \times G$: 
\begin{equation}
\begin{split}
K^{\chi,c} := \bigg\{ \omega^{\chi,c} \in \calA_{+,1}^*:& \exists \omega\in K \text{ such  that } 
\lim_{L\ra\infty} \omega(D_{L}^{\chi,c}) > 0 \text{ exists, and } \\
& \omega^{\chi,c} = \wslim_{L\ra\infty} \frac{\omega(\ \cdot \ D_{L}^{\chi,c} )}{\omega(D_{L}^{\chi,c})} \bigg\}.
\end{split}
\end{equation}
\end{definition}

By Lemma \ref{lem:ribchargeinv}, $D_{L'}^{\chi,c}$ is a supported on the boundary.
It follows that if $L'>L$ and  $\omega(D_{L'}^{\chi,c}) > 0$ then 
\begin{equation}\label{eqn:Ksets}
\frac{ \omega( \ \cdot \ D_{L'}^{\chi,c})}{\omega(D_{L'}^{\chi,c})}\bigg|_{\calA_L} = 
\frac{\omega( D_{L'}^{\chi,c} \cdot \ D_{L'}^{\chi,c})}{\omega(D_{L'}^{\chi,c})}\bigg|_{\calA_L}. 
\end{equation} 
In particular, we have that $\omega(\ \cdot \ D_{L'}^{\chi,c})$ is a \emph{positive} linear functional, which is not a priori clear, and $\omega( H_L^{\epsilon,\mu} D_{L'}^{\chi,c}) = 0$.
Thus, by Lemma \ref{lem:gshambound}, $K^{\chi,c} \subset K$ is a subset of the set of infinite volume ground states.
The interpretation of a state in $K^{\chi,c}$ is that it has a \emph{global} excitation of type $(\chi,c)$, hence the projection onto the charge $(\chi,c)$ in the region $\Lambda_L$ has a positive expectation value as $L$ goes to infinity.
The assumption that $  \lim_{L\ra\infty} \omega(D_{L}^{\chi,c})$ exists is always satisfied, as follows from the next lemma.

\begin{lemma}\label{lem:asymptoticcoef}
The limit $ \lambda_{\chi,c}(\omega) := \lim_{L\ra\infty} \omega(D_{L}^{\chi,c})$ exists for all ground states $\omega$ and 
we have $ \lambda_{\chi,c}(\omega) \geq 0$.
Furthermore, if $\omega^{\chi,c}\in K^{\chi,c}$ then  $\lambda_{\sigma,d}(\omega^{\chi,c}) = \delta_{(\sigma,d),(\chi,c)}$.
\end{lemma}

\begin{proof}
Let $L''> L' > L$.
First, we claim that 
\begin{align}
\label{eqn:chargecont1} D_L^\chi(\calG_{L''}^{\epsilon,\mu}) &\subset D_{L'}^\chi(\calG_{L''}^{\epsilon,\mu}) & & \text{if } \chi \neq \iota\\
\label{eqn:chargecont2} D_L^c(\calG_{L''}^{\epsilon,\mu}) &\subset D_{L'}^c(\calG_{L''}^{\epsilon,\mu}) & & \text{if } c \neq e\\ 
\label{eqn:chargecont3} D_{L'}^\iota(\calG_{L''}^{\epsilon,\mu}) &\subset D_{L}^\iota(\calG_{L''}^{\epsilon,\mu}) & &\text{if } \chi = \iota\\ 
\label{eqn:chargecont4} D_{L'}^e(\calG_{L''}^{\epsilon,\mu}) &\subset D_{L}^e(\calG_{L''}^{\epsilon,\mu}) & &\text{if } c = e 
\end{align}
(see Lemma~\ref{lem:gsspan} for a description of $\calG_{L''}^{\epsilon,\mu}$).
Note the reversal of $L$ and $L'$ in the last two equations.
The reason is that while in the first two equations, the \emph{presence} of a charge in the region is measured, in the last two equations it is the \emph{absence} of any charge that is important.

To see why these equations are true, consider first $\chi \neq \iota$ and note that the subspace $D_{L}^\chi(\calG_{L''}^{\epsilon,\mu})$ is spanned by 
$\{ F_{\rho}^{\chi,e}F_{\sigma}^{\iota,c}  \Omega:\forall \Omega\in \calG_{L''}; c\in G; \rho, \sigma$ paths connecting sites from the interior of $ \Lambda_{L''}$ to the boundary such that $\partial_0\rho \subset \Lambda_{L}  \}$.
The same statement is true if we replace $L$ by $L'$.
Thus, $D_L^\chi$ selects for a $\chi$-excitation in the region $\Lambda_L$ 
whereas $D_{L'}^\chi$ selects for a $\chi$-excitation in the region $\Lambda_{L'}$.
The later condition is weaker.
Therefore, $D_{L}^\chi|_{\calG_{L''}^{\epsilon,\mu}} \leq D_{L'}^\chi|_{\calG_{L''}^{\epsilon,\mu}} $ as projections.
A similar argument gives $D_{L}^c|_{\calG_{L''}^{\epsilon,\mu}} \leq D_{L'}^c|_{\calG_{L''}^{\epsilon,\mu}} $ as projections.
If $\chi=\iota$, $D_{L'}^{\iota}$ selects for a trivial $\epsilon$-type charge (i.e., the absence of an electric charge) in the region $\Lambda_{L'}$ 
while $D_{L}^{\iota}$ selects for a trivial $\epsilon$-type charge in the region $\Lambda_{L}$.
The later condition is weaker.  
Therefore, $D_{L'}^\iota|_{\calG_{L''}^{\epsilon,\mu}} \leq D_{L}^\iota|_{\calG_{L''}^{\epsilon,\mu}} $ as projections.
A similar argument gives $D_{L'}^e|_{\calG_{L''}^{\epsilon,\mu}} \leq D_{L}^e|_{\calG_{L''}^{\epsilon,\mu}} $ as projections. 
This shows that~\eqref{eqn:chargecont1}--\eqref{eqn:chargecont4} hold.

Let $\omega \in K$ be an infinite volume ground state.
As remarked below Definition~\ref{def:chargedgs}, $\omega|_{\calA_{L''}} \in K_{L''}$ is a ground state functional for $H_{L''}^{\epsilon,\mu}$.
Consider the sequence $\{ \omega(D_L^{\chi}) \}_{L=2}^\infty$.
If $\chi \neq \iota$, the inclusion~\eqref{eqn:chargecont1} gives that $ \omega( D_{L'}^\chi - D_{L}^\chi) \geq 0$,
thus the sequence is increasing.
The sequence is also bounded, $\omega(D_L^\chi) \leq \| D_L^\chi \| = 1$.
Hence we have a uniformly bounded and increasing sequence, and therefore the limit $ \lim_{ L\ra \infty} \omega(D_L^\chi)$ exists.
A similar argument gives that the limit $\lim_{L\ra \infty} \omega(D^{\chi,c}_L)$ exists if $\chi\neq \iota$ and $c \neq e$.

If $\chi \neq \iota$ and $c =e$, where there is a non-trivial electric charge and the magnetic charge is trivial, we can rewrite the projector $D_L^{\chi,e}$ as 
\begin{equation*}
D_L^{\chi, e} = D_L^{\chi} D_L^e =  D_L^{\chi} \bigg( I - \sum_{c \neq e} D_L^c \bigg)  = D_L^{\chi} - \sum_{c \neq e} D_L^{\chi,c}.
\end{equation*}
This is enough to show the limit $\lim_{L\ra \infty} \omega(D^{\chi,e}_L)$ exists.
By similar arguments, 
the limits exist for the cases $\chi = \iota$ with $c \neq e$, and when $\chi = \iota$ with $ c = e$.
The limits are always positive, since $\omega(D_L^{\chi,c}) \geq 0$ for all $L$.

To prove the second claim,
let $\omega^{\chi,c}\in K^{\chi,c}$ and choose $\omega\in K$ such that
$$
\omega^{\chi,c} = \wslim_{L'\ra\infty} \frac{\omega(\ \cdot \ D_{L'}^{\chi,c} )}{\omega(D_{L'}^{\chi,c})}.
$$
We use freely that the charge projectors commute.
Equations \eqref{eqn:chargecont1}--\eqref{eqn:chargecont2} imply that if $\chi \neq \iota$ and $c \neq e$ then $\omega(D_L^{\chi,c}D_{L'}^{\chi,c}) = \omega(D_L^{\chi,c})$.
Equations \eqref{eqn:chargecont3}--\eqref{eqn:chargecont4} imply that $\omega(D_L^{\iota,e} D_{L'}^{\iota,e}) = \omega(D_{L'}^{\iota,e})$.
If $\chi \neq \iota$, \eqref{eqn:chargecont1} and \eqref{eqn:chargecont3} imply that $\omega(D_L^{\chi,e} D_{L'}^{\chi,e})  = \omega(D_L^\chi D_{L'}^e)$. 
From the Cauchy-Schwarz inequality, it follows that
\begin{align*}
\abs{ \omega(D_L^{\chi,e} D_{L'}^{\chi,e}) - \omega(D_L^{\chi,e} )} 
& = \abs{ \omega( D_L^\chi (D_{L'}^e - D_L^e)}\\
& \leq \sqrt{\omega(D_L^\chi )}\sqrt{\abs{ \omega(( D_{L'}^e - D_L^e)^2) }}\\
&= \sqrt{\omega(D_L^\chi )}\sqrt{\abs{ \omega( D_{L'}^e + D_L^e - 2 D_L^e D_{L'}^e)}}\\
&= \sqrt{\omega(D_L^\chi )}\sqrt{\abs{ \omega( D_{L'}^e - D_L^e) }}
\end{align*}
and similarly for $c \neq e$, $\abs{ \omega(D_L^{\iota,c} D_{L'}^{\iota,c}) - \omega(D_L^{\iota,c} )} \leq \sqrt{\omega(D_L^c)}\sqrt{\abs{\omega(D_{L'}^\iota - D_L^\iota)}}$.

Let $\epsilon >0$ be given.  
The previous arguments show that for all $(\chi,c) \in \widehat{G}\times G$, there exists $l$ such that if  $L'>L>l$ then
\begin{equation}\label{eqn:fullcharge}
\abs{ \omega(D_L^{\chi,c}D_{L'}^{\chi,c}) - \omega(D_{L}^{\chi,c})} < \epsilon.
\end{equation}
Thus,
\begin{align*}
\lambda^{\chi,c}(\omega^{\chi,c})
& = \lim_{L\ra\infty} \lim_{L'\ra\infty}  \frac{\omega(D_L^{\chi,c} D_{L'}^{\chi,c})}{\omega(D_{L'}^{\chi,c})}\\
&=\lim_{L\ra\infty} \lim_{L'\ra\infty} \frac{\omega( D_{L}^{\chi,c})}{\omega(D_{L'}^{\chi,c})}\\
&=1.
\end{align*}
Therefore, combining $\sum_{\sigma,d} \lambda_{\sigma,d}(\omega^{\chi,c}) = 1$ 
and $\lambda_{\chi,c}(\omega^{\chi,c}) =1$ 
gives $\lambda_{\sigma,d}(\omega^{\chi,c}) = \delta_{(\sigma,d),(\chi,c)}$.
\end{proof}

From the arguments given in the previous lemma, we can achieve a slightly stronger bound which will be used later.
Let $\epsilon>0$ be given. 
Then, for all $(\chi,c)\in \widehat{G} \times G$, we show that there exists $l$ such that if $L' > L >l$ then 
\begin{equation}\label{eqn:chargemoment}
\abs{ \omega\left( ( D_{L'}^{\chi,c} - D_{L}^{\chi,c})^2\right) } <\epsilon.
\end{equation}
If $\chi \neq \iota$ and $ c\neq e$, \eqref{eqn:chargecont1} and \eqref{eqn:chargecont2} give that $\omega\left( ( D_{L'}^{\chi,c} - D_{L}^{\chi,c})^2\right) = \omega( D_{L'}^{\chi,c} - D_{L}^{\chi,c})$. 
Similarly, we find that $\omega\left( ( D_{L'}^{\iota,e} - D_{L}^{\iota,e})^2\right) = \abs{\omega( D_{L'}^{\iota,e} - D_{L}^{\iota,e})}$.
If $\chi \neq \iota$ then 
\begin{align*}
\abs{ \omega\left( ( D_{L'}^{\chi,e} - D_{L}^{\chi,e})^2\right) }
& = \abs{\omega( D_{L'}^{\chi,e} + D_L^{\chi,e} - 2 D_L^{\chi,e}D_{L'}^{\chi,e})}\\
& \leq \abs{\omega( D_{L'}^{\chi,e}  - D_L^\chi D_{L'}^e)} +  \abs{\omega( D_L^{\chi,e} - D_L^\chi D_{L'}^e)}\\
& = \abs{\omega\left( D_{L'}^e (D_{L'}^\chi -D_L^\chi)\right)} + \abs{ \omega( D_L^{\chi} (D_{L'}^e - D_L^e))},
\end{align*}
with a similar bound holding if $ \chi = \iota$ and $c \neq e$.  Thus, \eqref{eqn:chargemoment} holds.

Lemma \ref{lem:asymptoticcoef} allows us to distinguish ground states with different charges and makes it possible to decompose any ground state into charged ground states. This is the first part of our main result.
\begin{theorem}\label{thm:qdoubgs1}
Let $\omega\in K$ be a ground state.
Then there exists a convex decomposition of $\omega$  as
\begin{equation}\label{eqn:gsdecomp4}
\omega = \sum_{\chi\in \widehat{G},c \in G} \lambda_{\chi,c}(\omega) \omega^{\chi,c} \qquad \text{ where } \quad \omega^{\chi,c}\in K^{\chi,c}.
\end{equation}
Furthermore, we can calculate the coefficients explicitly as
\begin{equation}
\lambda_{\chi,c}(\omega) = \lim_{L\ra\infty} \omega(D_{L}^{\chi,c}).
\end{equation}
If $ \lambda_{\chi,c}(\omega) >0$ then 
\begin{equation}
	\label{eqn:wslimit}
\omega^{\chi,c} = \wslim_{L\ra\infty} \frac{\omega(\ \cdot \ D_{L}^{\chi,c} )}{\omega(D_{L}^{\chi,c})}.
\end{equation}
\end{theorem}

\begin{proof}
For convenience, in this proof we denote $\lambda_{\chi,c}(\omega) = \lambda_{\chi,c}$.

By Lemma \ref{lem:asymptoticcoef}, the values $\lambda_{\chi,c} \geq 0$  are well-defined,
so we will have to show that the limit in equation~\eqref{eqn:wslimit} exists and that the decomposition in equation~\eqref{eqn:gsdecomp4} agrees with the state $\omega$.

Let $L' > L \geq 2$.  
Since for each $L'$ the charge projections add up to the identity, by equation~\eqref{eqn:localprojcomplete}, we have $\omega = \sum_{\chi,c} \omega( \ \cdot \ D_{L'}^{\chi,c})$ for all $L'$.
By Lemma \ref{lem:gshambound}, $\omega|_{\calA_L}$
is a finite volume ground state for $H_L^{\epsilon,\mu}$ for all $L\geq 2$ (see also the remark after Definition~\ref{def:chargedgs}).

Now suppose $\lambda_{\chi,c}>0$.
Let $\epsilon >0$ be given and suppose $\epsilon$ is small enough such that $\lambda_{\chi,c} > \epsilon >0$. 
By Lemma~\ref{lem:asymptoticcoef}, and inequalities~\eqref{eqn:fullcharge} and~\eqref{eqn:chargemoment}, there exists $L>0$ such that if $L''> L' >L$ then
\[
\abs{\omega( D_{L'}^{\chi,c}) - \lambda_{\chi,c}}  <   \epsilon, \quad
\abs{ \omega( D_{L''}^{\chi,c}-  D_{L'}^{\chi,c})} < \epsilon, \quad \text{ and }\quad
\abs{ \omega\left( ( D_{L''}^{\chi,c} -  D_{L'}^{\chi,c})^2\right) } < \epsilon.
\]
We also demand that $\omega(D^{\chi,c}_{L'}) > 0$ for all $L' > L$, which can always be achieved by choosing $L$ big enough.
Note that $\lambda_{\chi,c} \leq 1$, so we can restrict to $\epsilon < 1$.
Let $A\in \calA_{L}$, then 
\begin{align}\label{eqn:cauchyproj}
&\left\vert \frac{\omega(A D_{L'}^{\chi,c} )}{\omega(D_{L'}^{\chi,c})}  - \frac{\omega(A D_{L''}^{\chi,c} )}{\omega(D_{L''}^{\chi,c})} \right\vert \\
& \quad\quad\quad\quad =  \frac{1}{\omega(D_{L'}^{\chi,c}) \omega(D_{L''}^{\chi,c})}
\left\vert  \omega(A D_{L'}^{\chi,c} )\omega(D_{L''}^{\chi,c}) -  \omega(A D_{L''}^{\chi,c} )\omega(D_{L'}^{\chi,c}) \right\vert\\
& \quad\quad\quad\quad \leq \frac{1}{\omega(D_{L'}^{\chi,c}) \omega(D_{L''}^{\chi,c})}
\Big( \left| \omega(A D_{L'}^{\chi,c} )\right| \left\vert   \omega(D_{L''}^{\chi,c}) - \omega( D_{L'}^{\chi,c} ) \right\vert \\
& \qquad \qquad \qquad \qquad +
\omega(D_{L'}^{\chi,c})\abs{\omega(A (D_{L'}^{\chi,c} - D_{L''}^{\chi,c}) )}\Big). 
\end{align}
Recall that by Lemma \ref{lem:ribchargeinv},
 we have that $\omega(A D_{L'}^{\chi,c}) = \omega(D_{L'}^{\chi,c} A D_{L'}^{\chi,c})$.
It follows that $\abs{\omega(A D_{L'}^{\chi,c})} \leq \|A\| \omega( D_{L'}^{\chi,c})$.
We also note the estimate $\frac{1}{\omega(D_{L''}^{\chi,c})} \leq \frac{1}{\lambda_{\chi,c} - \epsilon } $.
The last term can be estimated using the Cauchy-Schwarz inequality,
\[
	\left| \omega(A(D_{L'}^{\chi,c} - D_{L''}^{\chi,c}) \right| 
	\leq \sqrt{\omega(A^*A)} \sqrt{\abs{\omega\left( (D_{L'}^{\chi,c} - D_{L''}^{\chi,c})^2\right)}}
	\leq \|A\| \sqrt{\epsilon}.
\]
Combining these estimates we obtain the bound
\begin{equation}\label{eqn:cauchyseq}
	\left\vert \frac{\omega(A D_{L'}^{\chi,c} )}{\omega(D_{L'}^{\chi,c})}  - \frac{\omega(A D_{L''}^{\chi,c} )}{\omega(D_{L''}^{\chi,c})} \right\vert \leq \frac{\|A\| (\epsilon + \sqrt{\epsilon})}{\lambda_{\chi,c} - \epsilon}.
\end{equation}
Because $\lambda_{\chi,c} > 0$, this goes to zero as $\epsilon$ goes to zero.

Thus for each pair $(\chi,c)$ there is a sequence of states
\begin{equation}
\omega_L^{\chi,c} = \frac{\omega( \ \cdot \ D_{L}^{\chi,c})}{\omega(D_{L}^{\chi,c})} ,
\end{equation}
converging in the weak$^*$ limit (or $\omega^{\chi,c}$ is the zero functional if $\lambda_{\chi,c} = 0$).
We have the following properties (cf.\ equations (4.48)--(4.50) in~\cite{KomaN}):
\begin{align}
&\wslim_{L\ra\infty} \omega_L^{\chi,c}  := \omega^{\chi,c} \quad \text{ exists}, \\
&\wslim_{L\ra \infty} \left\vert \omega - \sum_{\chi,c} \lambda_{\chi,c}\omega_L^{\chi,c} \right\vert = 0\\
& \omega^{\chi,c}_{L'}( H_{L}^{\epsilon,\mu}) = 0  \quad \text{ for all } L'> L\geq 2.
\end{align}
The last property follows directly from the fact that $\omega_{L'}^{\chi,c}\in K_{L}$ for all $L'>L$
and $\omega^{\chi,c}$ is an infinite volume ground state for all $(\chi,c)$.
Thus we have proven the ground state decomposition as in equation~\eqref{eqn:gsdecomp4}.
\end{proof}

\begin{corollary}\label{cor:face}
For all $(\chi,c)\in \widehat{G}\times G$, $K^{\chi,c}$ is a face in the set of all states.
In particular, if $\omega^{\chi,c} \in K^{\chi,c}$ is an extremal point of $K^{\chi,c}$ 
then $\omega^{\chi,c}$ is a pure state.
\end{corollary}

\begin{proof}
Let $\omega^{\chi,c} \in K^{\chi,c}$ and suppose $\phi \leq \lambda \omega^{\chi,c}$. 
Since the set of ground states is a face
this implies $\phi$ is a ground state.
By Theorem \ref{thm:qdoubgs1}, decompose $\phi$ as 
$\phi = \sum \lambda_{\sigma,d}(\phi)  \phi^{\sigma, d}$.
By Lemma \ref{lem:asymptoticcoef}, 
\begin{align*}
\lambda_{\sigma,d}(\phi) &= \lim_{L\ra \infty} \phi(D_L^{\sigma,d}) \\
& \leq \lambda \lim_{L\ra\infty}  \omega^{\chi,c}(D^{\sigma,d}_L)\\
& = \lambda \delta_{(\sigma,d), (\chi,c)}.
\end{align*}
Therefore, $\phi = \phi^{\chi,c} \in K^{\chi,c}$, which shows $K^{\chi,c}$ is a face
in the set of all states.

Suppose $\omega^{\chi,c}$ is an extremal state of $K^{\chi,c}$
and that $\omega^{\chi,c}$ can be decomposed as a convex combination of states
\[ \omega^{\chi,c}  = c_1 \omega_1 + c_2 \omega_2.\]
It follows that $c_i \omega_i \leq \omega^{\chi,c}$, 
so by the face property of $K^{\chi,c}$, $\omega_i \in K^{\chi,c}$.
The supposition that the state $\omega^{\chi,c}$ is extremal in $K^{\chi,c}$ leads to the conclusion that $\omega_1 = \omega_2 =\omega^{\chi,c}$. 
\end{proof}

We note that the definition we take for a face does not require it to be a closed set in the weak$^*$ topology (but see also Theorem~\ref{thm:weakclosure}).
The decomposition above suggests that each ground state can be decomposed into ground states that are related to the superselection sectors.
Indeed, the \emph{pure} states in $K^{\chi,c}$ are equivalent to one of the charged states that we constructed before.
Two states $\omega_1$ and  $\omega_2$ are said to be \emph{equivalent} if their corresponding GNS representations are unitarily equivalent.

\begin{theorem}\label{thm:eqstates}
If $\omega^{\chi,c} \in K^{\chi,c}$ is a pure state then $\omega^{\chi,c}$ 
is equivalent to a single excitation ground state $\omega_s^{\chi,c}$,
as defined in equation~\eqref{eqn:singleexcitation}. 
\end{theorem}
\begin{proof}
First, we notice that $\omega_s^{\chi,c} \in K^{\chi,c}$ for all sites $s$ since 
$ \frac{\omega_s^{\chi,c}( A D_{L'}^{\chi,c})}{\omega_s^{\chi,c}(D_{L'}^{\chi,c})} = \omega_s^{\chi,c}(A)$
for all $A \in \calA_L$ and $L'>L$, where $L$ is chosen large enough so that $ s \in \mathcal{S}_L$.

Let $\omega^{\chi,c} \in K^{\chi,c}$ and 
let $\omega \in K$ be a ground state such that we have
$\omega^{\chi,c} = \wslim \frac{\omega( \cdot D^{\chi,c}_L)}{\omega(D_L^{\chi,c})}$ and $\lambda_{\chi,c}>0$.
Let $\epsilon >0$ be given and suppose $\epsilon$ is small enough such that $\lambda_{\chi,c} > \epsilon >0$. 
By Lemma \ref{lem:asymptoticcoef} and equation \eqref{eqn:fullcharge}, there is an $L>0$  such that if $L''\geq L' \geq L$ then $\abs{ \omega(D^{\chi,c}_{L''}) - \omega(D_{L'}^{\chi,c}D^{\chi,c}_{L''})}  < \epsilon$ and $\abs{ \omega(D^{\chi,c}_{L''}) - \omega(D_{L'}^{\chi,c})} < \epsilon$.
Fix an operator $ A \in \calA_{loc} \cap \calA_{(L)^c}$.
Then, there is an $L'>L+1$ such that $A$ is supported on the annulus $A \in \calA_{L'-2} \cap \calA_{(L)^c}$,
and an $L''> L'+1$ such that 
\begin{equation}\label{eqn:bound1}
\abs{ \omega^{\chi,c}( A) - \frac{\omega(A D^{\chi,c}_{L''})}{\omega(D^{\chi,c}_{L''})}} < \|A\|\epsilon.
\end{equation}
The estimate \eqref{eqn:cauchyseq} also holds.
Applying \eqref{eqn:fullcharge} and the estimates 
$$
\frac{1}{\omega(D_{L'}^{\chi,c})} \leq \frac{1}{\lambda_{\chi,c} - \epsilon } \quad \mbox{ and } \quad
\frac{1}{\omega(D_L^{\chi,c}D_{L'}^{\chi,c})} \leq \frac{1}{\lambda_{\chi,c} - \epsilon }
$$
gives
\begin{align*}
\abs{ \frac{\omega(A D_{L'}^{\chi,c})}{\omega(D_{L'}^{\chi,c})} - \frac{\omega(A D_{L}^{\chi,c} D_{L'}^{\chi,c})}{\omega(D_{L}^{\chi,c}D_{L'}^{\chi,c})}} 
&\leq  \abs{  \frac{\omega(A D_{L'}^{\chi,c})}{\omega(D_{L'}^{\chi,c})} - \frac{\omega(A D_{L}^{\chi,c} D_{L'}^{\chi,c})}{\omega(D_{L'}^{\chi,c})}} \\
&\quad\quad\quad+ \abs{ \frac{\omega(A D_{L}^{\chi,c} D_{L'}^{\chi,c})}{\omega(D_{L'}^{\chi,c})}- \frac{\omega(A D_{L}^{\chi,c} D_{L'}^{\chi,c})}{\omega(D_{L}^{\chi,c}D_{L'}^{\chi,c})}} \\
&\leq \frac{\|A\| (\epsilon + \sqrt{\epsilon})}{(\lambda_{\chi,c} - \epsilon)^2}
\end{align*}
Combining the above estimates, it follows that for a site $s$ contained in $\Lambda_L$
\begin{equation}\label{eqn:eqstates1}
\begin{split}
& \abs{ \omega^{\chi,c}(A) - \omega_s^{\chi,c}(A) } \leq 
\abs{\omega^{\chi,c}(A) - \frac{\omega(A D^{\chi,c}_{L''})}{\omega(D^{\chi,c}_{L''})}}
+ \abs{ \frac{\omega(A D^{\chi,c}_{L''})}{\omega( D^{\chi,c}_{L''})} -  \frac{\omega(A D_{L'}^{\chi,c})}{\omega(D_{L'}^{\chi,c})}} \\
& \, + \abs{ \frac{\omega(A D_{L'}^{\chi,c})}{\omega(D_{L'}^{\chi,c})} - \frac{\omega(A D_{L}^{\chi,c} D_{L'}^{\chi,c})}{\omega(D_{L}^{\chi,c}D_{L'}^{\chi,c})}}
+\abs{ \frac{\omega(A D_{L}^{\chi,c} D_{L'}^{\chi,c})}{\omega(D_{L}^{\chi,c}D_{L'}^{\chi,c})} -  \frac{ \omega(D_{L'}^{\chi,c} D_{L}^{\chi,c} A D_{L}^{\chi,c}D_{L'}^{\chi,c})}{\omega(D_{L}^{\chi,c}D_{L'}^{\chi,c})}}\\
& \qquad \qquad + \abs{  \frac{ \omega(D_{L'}^{\chi,c} D_{L}^{\chi,c} A D_{L}^{\chi,c}D_{L'}^{\chi,c})}{\omega(D_{L}^{\chi,c}D_{L'}^{\chi,c})} - \omega_s^{\chi,c}(A) } \\
& \quad\quad\quad \leq \|A\|\epsilon  +  2\frac{\|A\| (\epsilon + \sqrt{\epsilon})}{(\lambda_{\chi,c} - \epsilon)^2}
+  \abs{  \frac{ \omega(D_{L'}^{\chi,c}D_{L}^{\chi,c}A D_{L}^{\chi,c}D_{L'}^{\chi,c})}{\omega(D_{L}^{\chi,c}D_{L'}^{\chi,c})} - \omega_s^{\chi,c}(A) },
\end{split}
\end{equation}
where the fourth term vanishes by Lemma \ref{lem:ribchargeinv}.
The last term will be shown to be identically zero.

Recall that $\omega_s^{\chi,c}|_{\calA_{L'}} = \omega^0( F_{\rho_{L'}}^{\chi,c *} \ \cdot \ F_{\rho_{L'}}^{\chi,c})|_{\calA_{L'}}$ 
where $\rho_{L'}$ is a ribbon that connects the site $s$ to the boundary of $\Lambda_{L'}$.
Denote the subspace $\calG_{L',L}^{\chi,c} := D_{L}^{\chi,c} D_{L'}^{\chi,c} (\calG_{L'}^{\epsilon,\mu})$;
it is spanned by simple vectors of the form 
$F_{\sigma}^{\chi,e}F_{\sigma'}^{\iota, c} \Omega$ where $ \Omega \in \calG_{L'}$, and $\sigma$ and $\sigma'$ are a ribbons connecting sites in 
$\mathcal{S}_L$ to the boundary of $\Lambda_{L'}$, see Lemma \ref{lem:gsspan}.

First, we consider the case $\psi \in\calG_{L',L}^{\chi,c}$ and $\psi =  F_\sigma^{\chi,c} \Omega$.
Indeed, for each $\sigma$ as above, we can find a new ribbon, $\sigma' = \sigma_1 \sigma_2 \sigma_3$, see Figure~\ref{fig:eqstates}, connecting 
$y$ to the boundary of $ \Lambda_{L'}$ with the following properties 
\begin{align}\label{eqn:pathdecomp}
&F_{\sigma_1}^{\chi,c}\in \calA_L \text{ and } \sigma_1\cap \Lambda_L =\sigma \cap \Lambda_L\\
&F_{\sigma_2}^{\chi,c} \in \calA_{L'-1}\cap \calA_{(L+1)^c} \text{ and }  
\sigma_2 \cap \Lambda_{L'-2}\cap \Lambda_{(L+2)^c}  = \rho_{L'} \cap  \Lambda_{L'-2}\cap \Lambda_{(L+2)^c} \\
&F_{\sigma_3}^{\chi,c} \in \calA_{L'} \cap  \calA_{( L'-1)^c} 
\text{ and } \sigma_3\cap \Lambda_{L''} \cap  \Lambda_{( L'-1)^c} =\sigma \cap \Lambda_{L''} \cap  \Lambda_{( L'-1)^c}\\
&F_\sigma^{\chi,c} \Omega = F_{\sigma_1}^{\chi,c} F_{\sigma_2}^{\chi,c} F_{\sigma_3}^{\chi,c} \Omega. \label{eqn:pathdecomp2}
\end{align} 
Here we used that the state only depends on the endpoints of the ribbon, not on the path.
Decompose 
$\psi = F_{\sigma}^{\chi,c} \Omega
= F_{\sigma_{1}}^{\chi,c} F_{\sigma_{2}}^{\chi,c} F_{\sigma_{3}}^{\chi,c}\Omega $.

\begin{figure}
\begin{center} 
\includegraphics[width=.5\textwidth]{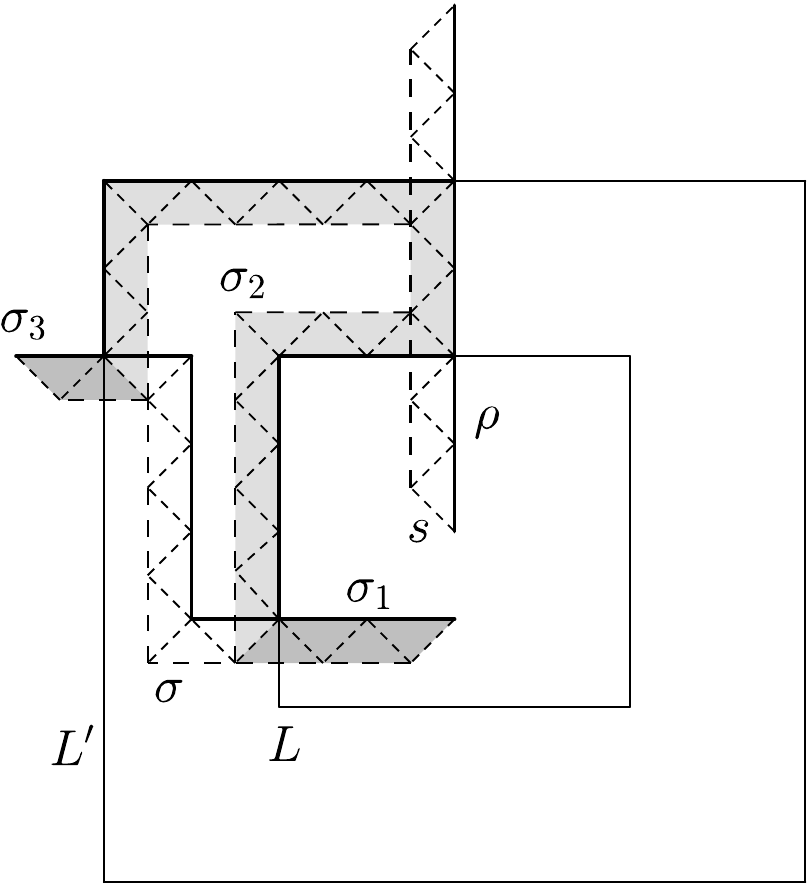}
\end{center}
\caption{A depiction of the ribbons $\rho$, $\sigma$ and $\sigma'=\sigma_1\sigma_2\sigma_3$. 
The ribbon $\sigma'$ is shaded, $\sigma_2$ is distinguished by a lighter shade.}
\label{fig:eqstates}
\end{figure}

Suppose $A$ is a product of ribbon operators.
If $A$ is not a product of closed paths, 
then its action on $\calG_{L',L}^{\chi,c}$ yields a subspace with strictly higher energy with respect to $H_{L''}^{\chi,c}$.
Thus, if $\psi \in \calG_{L',L}^{\chi,c}$ then $\langle \psi, A \psi \rangle = 0 = \omega_s^{\chi,c}(A)$.
If $A$ is a product of closed paths, then $A$ leaves the frustration free ground state invariant,
 $A\Omega = \Omega$.
Let $k\in\CC$  be such that $A F_{\rho_{L'}}^{\chi,c} = kF_{\rho_{L'}}^{\chi,c} A$;
$k$ can be computed from ribbon intertwining relations \eqref{eqn:ribbonrelation}.
It follows that 
\begin{align*}
\omega_s^{\chi,c}(A)& = \langle F_{\rho_{L'}}^{\chi,c} \Omega, A F_{\rho_{L'}}^{\chi,c} \Omega \rangle\\
& = k \langle \Omega, A \Omega \rangle \\
& = k.
\end{align*}
Now going back to the properties of $\sigma_{l}$ in equations \eqref{eqn:pathdecomp}--\eqref{eqn:pathdecomp2}, we have that 
\begin{equation}
[A, F_{\sigma_1}^{\chi,c}] = 0, \quad [A, F_{\sigma_3}^{\chi,c}] = 0, 
\quad \text{ and } \quad A F_{\sigma_2}^{\chi,c} = kF_{\sigma_2}^{\chi,c} A .
\end{equation}
Thus, if $A$ is a product of closed ribbon operators then
\begin{align*}
A \psi & =  A F_{\sigma_{1}}^{\chi,c} F_{\sigma_{2}}^{\chi,c} F_{\sigma_{3}}^{\chi,c}\Omega \\
&=  k F_{\sigma_{1}}^{\chi,c} F_{\sigma_{2}}^{\chi,c} F_{\sigma_{3}}^{\chi,c}A \Omega \\
& = k  \psi
\end{align*}
and $\psi$ has eigenvalue $k$.
For the case $\psi = F_{\sigma}^{\chi,e}F_{\tau}^{\iota,c}\Omega$,
the decomposition $\sigma = \sigma_1 \sigma_2 \sigma_3$ and $ \tau = \tau_1 \tau_2 \tau_3$ 
as above then we can choose $ \sigma_2$ and $\tau_2$ to coincide on the annulus $\Lambda_{L'-2}\cap \Lambda_{(L+2)^c}$.
Therefore, the same argument as above shows $A \psi = k \psi$.

For a general $\psi \in\calG_{L',L}^{\chi,c}$, $\psi$ is a linear combination of the simple vectors $ F_\sigma^{\chi,e} F_{\sigma'}^{\iota,c}\Omega$.
Thus, by linearity $A\psi = k \psi$ for all $\psi \in\calG_{L',L}^{\chi,c}$.
Therefore, if $\psi$ is normalized 
\begin{equation}
\langle \psi, A \psi \rangle = k = \omega_s^{\chi,c}(A).
\end{equation}
Note that we already established this equation for $A$ an open ribbon operator.

Since ribbon operators span the algebra $\calA_{L'-2} \cap \calA_{(L+2)^c}$, 
we extend the above argument by linearity 
so that 
\begin{equation}\label{eqn:eqstate3}
\langle \psi, A \psi \rangle  = \omega_s^{\chi,c}(A)\quad \text{ for all } \quad A \in \calA_{L'-2} \cap \calA_{(L+2)^c}.
\end{equation} 

A general mixed state supported on $\calG_{L',L}^{\chi,c}$ is of the form 
$$ 
\phi = \frac{ \omega_{L'}(D_{L'}^{\chi,c} D_{L}^{\chi,c} \ \cdot \ D_{L}^{\chi,c}D_{L'}^{\chi,c})}{\omega_{L'}(D_{L}^{\chi,c}D_{L'}^{\chi,c})} =
\sum c_{\psi} \langle \psi, A \psi\rangle,
$$ 
where $ \omega_{L'} \in K_{L'}$ and each $\psi$ is a linear combination of vectors of the form $F_{\sigma}^{\chi,e}F_\xi^{\iota,c}\Omega$.
Since the $c_\psi$ add up to one, it follows that $ \phi(A) = \omega_s^{\chi,c}(A)$ for all $A \in \calA_{L'-2} \cap \calA_{(L+2)^c}$.

For the ground state $\omega$, Lemma \ref{lem:gshambound} gives $\omega|_{\calA_{L'}} \in K_{L'}$.
Therefore,
\begin{equation}
\frac{ \omega(D_{L'}^{\chi,c}D_{L}^{\chi,c}A D_{L}^{\chi,c}D_{L'}^{\chi,c})}{\omega(D_{L}^{\chi,c}D_{L'}^{\chi,c})} = \omega_s^{\chi,c}(A) \quad \quad \text{ for all } A \in \calA_{L'-2} \cap \calA_{(L+2)^c}.
\end{equation}

Since $L'$ was chosen such that $L'>L$, and otherwise arbitrary, the estimate in \eqref{eqn:eqstates1} becomes
\begin{equation}
\abs{ \omega^{\chi,c}(A) - \omega_s^{\chi,c}(A) } \leq  \|A\| \left(\epsilon +2\frac{(\epsilon + \sqrt{\epsilon})}{(\lambda_{\chi,c} - \epsilon)^2}\right) \qquad \text{ for all } \quad A \in \calA_{loc}\cap \calA_{L^c}.
\end{equation}

Now suppose further that $\omega^{\chi,c}$ is a pure state.
Proposition \ref{prop:singleexc} also gives that the states $\omega_s^{\chi,c}$ are pure states.
Therefore, applying the criterion for equivalence of pure states (Corollary 2.6.11, \cite{BratteliR})
we have 
\begin{equation}\label{eqn:eqstates2}
\omega^{\chi,c} \approx \omega_s^{\chi,c} \text{ for all } (\chi,c) \in \widehat{G}\times G.
\end{equation}
This completes the proof.
\end{proof}

The above two theorems give a complete characterization of the ground states of the quantum double model.
The sets of states $K^{\chi,c}$ played an important role in the analysis.
We end our discussion by finding the weak$^*$ closure of these sets in the set of all states.
Recall that each state in $K^{\chi,c}$ has a charge $(\chi,c)$.
Now consider a sequence of states where the $\chi$ charge is gradually moved off to infinity.
The resulting weak$^*$ limit will be a state with only a charge $c$,
so the weak closures of the sets $K^{\chi,c}$ are strictly larger than $K^{\chi,c}$ (unless $\chi = \iota$ and $c = e$).
The final result is that this procedure of moving away charges gives the  weak$^*$ closures of the sets of charged ground states.

\begin{theorem}\label{thm:weakclosure}
The closures in the weak$^*$ topology are given by
\begin{equation}\label{eqn:sectclose}
 \overline{K^{\chi,c}}^w = \conv{K^0\cup K^{\chi,e} \cup K^{\iota,c} \cup K^{\chi,c}},
\end{equation}
where $\operatorname{Conv}$ denotes the convex hull.
\end{theorem}
\begin{proof}

First, we show that $K^{\chi,e} \subset \overline{K^{\chi,c}}^w$.

Let $\omega^{\chi,e} \in K^{\chi,e}$ and $\rho$ be a path extending to infinity based 
at site $s$. Consider the automorphism $\alpha_\rho^{\iota, c}$ (see \eqref{eqn:chargemorp}) 
that generates a charge of type $(\iota,c)$ located at the site $s$.
We claim that the state $\omega^{\chi,e} \circ \alpha_\rho^{\iota,c} \in K^{\chi,c}$.
To see this, from Theorem \ref{thm:qdoubgs1} and Lemma \ref{lem:asymptoticcoef},
write $ \omega^{\chi,e}(A) = \lim_{L\ra \infty} \omega^{\chi,e}( A D_{L}^{\chi,e} )/\omega^{\chi,e}(D_{L}^{\chi,e})$
for $A \in \calA$.
Notice that for $L'>L$ large enough such that $s \subset \Lambda_L$, equation \eqref{eqn:localprojribbonrelation1} gives
$F_{\rho_{L'}}^{\iota,c} D_L^{\chi,e} = D_L^{\chi,c}F_{\rho_{L'}}^{\iota,c}$
and $(\alpha_{\rho}^{\iota,c })^{-1}(D_L^{\chi,e}) = D_L^{\chi,c}$.
Thus,
\begin{equation}
\omega^{\chi,e} \circ \alpha_\rho^{\iota, c} (A) 
= \lim_{L\ra \infty} \frac{ \omega^{\chi,e}(  \alpha_\rho^{\iota, c}(A) D_{L}^{\chi,e} )}{\omega^{\chi,e}(D_{L}^{\chi,e})} 
= \lim_{L\ra \infty} \frac{ \omega^{\chi,e}\circ \alpha_\rho^{\iota, c}(A D_{L}^{\chi,c}) }{\omega^{\chi,e}\circ \alpha_\rho^{\iota, c} (D_{L}^{\chi,c})}.
\end{equation}
To finish the claim, we need to show $\omega^{\chi,e} \circ \alpha_\rho^{\iota, c} $ is a ground state.
Recall that $ D_{L'}^{\chi,e}\calG_{L'}^{\epsilon, \mu}$ is spanned by simple vectors of the form
$F_\sigma^{\chi,e} \Omega$ 
where $ \Omega \in \calG_{L'}$ and 
$\sigma$ is a ribbon connecting a site $s \in \mathcal{S}_{L'}$ to the 
boundary of $\Lambda_{L'}$.
Let $ \psi \in D_{L'}^{\chi,e}\calG_{L'}^{\epsilon, \mu}$ and write $ \psi = \sum_j b_j F_{\sigma_j}^{\chi,e}\Omega_j$.
Since $F_{\rho_{L'}}^{\iota,c} \psi = \sum_j b_j F_{\rho_{L'}}^{\iota,c} F_{\sigma_j}^{\chi,e}\Omega_j \in D_{L'}^{\chi,c}\calG_L^{\epsilon,\mu}$, it follows that 
$H_L^{\epsilon,\mu} F_{\rho_{L'}}^{\iota,c} \psi = 0$.
Indeed, we compute $\omega^{\chi,e}\circ\alpha_{\rho}^{\iota, c}$ is an infinite volume ground state:
\begin{align*}
\omega^{\chi,e}(\alpha_{\rho}^{\iota, c}( H_L^{\epsilon,\mu}))
&= \lim_{L'\ra\infty} \frac{ \omega^{\chi,e}(\alpha_{\rho}^{\epsilon,\mu}( H_L^{\epsilon,\mu} D_{L'}^{\chi,c}))}
{\omega^{\chi,e}\circ\alpha_{\rho}^{\iota, c}(D_{L'}^{\chi,c})} \\
&= \lim_{L'\ra\infty} \frac{ \omega^{\chi,e}(\alpha_{\rho}^{\epsilon,\mu}( D_{L'}^{\chi,c} H_L^{\epsilon,\mu} D_{L'}^{\chi,c}))}
{\omega^{\chi,e}\circ\alpha_{\rho}^{\iota, c}(D_{L'}^{\chi,c})}\\
&= \lim_{L'\ra\infty} \frac{  \omega^{\chi,e}( D_{L'}^{\chi, e} F_{\rho_{L'}}^{\iota,c *} H_{L}^{\epsilon,\mu}F_{\rho_{L'}}^{\iota,c} D_{L'}^{\chi,e} )}{\omega^{\chi,e}( D_{L'}^{\chi, e})}\\
& = 0,
\end{align*}
where the last equality is true since the state 
$$
\frac{  \omega^{\chi,e}( D_{L'}^{\chi, e} F_{\rho_{L'}}^{\iota,c *} \ \cdot \ F_{\rho_{L'}}^{\iota,c} D_{L'}^{\chi,e} )}{\omega^{\chi,e}( D_{L'}^{\chi, e})}
$$
is a mixed state supported on $D_{L'}^{\chi,c}\calG_L^{\epsilon,\mu}$.

Now consider a sequence $s_n$ of sites such that $s_1 = s$ and $ s_n \ra \infty$ as $n\ra \infty$.
Let $\rho_n$ be a ribbon extending to infinity based at the site $s_n$ 
and then define the sequence of states 
\[  \omega_{n}=  \omega^{\chi,e}\circ\alpha_{\rho_n}^{\iota,c} \in K^{\chi,c}.\]
For $A \in \calA_{loc}$, choose $n$ large enough so that
$ \alpha_{\rho_n}^{\iota, c}(A) = A$.
It follows that,
\begin{align*}
\omega_{n}( A ) & = \omega^{\chi,e}\circ\alpha_{\rho_n}^{\iota,c}(A)= \omega^{\chi, e}(A).
\end{align*}
Therefore, $\wslim \omega_n = \omega^{\chi, e} \in K^{\chi,e}$.
By similar arguments one can show the inclusion,
\[\conv{K^0\cup K^{\chi,e} \cup K^{\iota,c} \cup K^{\chi,c}} \subset  \overline{K^{\chi,c}}^w.\]

Now, to show the reverse inclusion, suppose $\widehat{\omega}^{\chi,c} \in \overline{K^{\chi,c}}^w$ 
and let $\omega^{\chi,c}_\lambda \in K^{\chi,c}$ be a net in $K^{\chi,c}$ such that $ \wslim_{\lambda}\omega_\lambda^{\chi,c} = \widehat{\omega}^{\chi,c}$.   
For each $ \lambda$, we can write
$$
\omega_\lambda^{\chi,c} = \lim_{L \ra \infty } \frac{\omega_\lambda( \ \cdot \ D_L^{\chi,c})}{\omega_\lambda(D_L^{\chi,c})} .
$$
The proof of Lemma \ref{lem:asymptoticcoef}
gives that $ D^{\sigma, d}_L D^{\chi,c}_{L'} |_{\calG_{L'}^{\epsilon,\mu}} = 0$ if $(\sigma, d)$ is not in the set $\{ (\chi,c), (\chi,e), (\iota,c), (\iota, e) \}  $.
Thus, in that case,
\begin{align}\label{eqn:gswslim}
\widehat{\omega}^{\chi,c}( D_L^{\sigma, d}) &= \lim_{\lambda } \omega_{\lambda}^{\chi,c}( D_L^{\sigma, d}) \\
& = \lim_{\lambda} \lim_{L'\ra\infty }\frac{\omega_{\lambda}^{\chi,c}(D_L^{\sigma, d} D_{L'}^{\chi,c} )}{\omega_\lambda^{\chi,c}(D_{L'}^{\chi,c})}\\
& = 0.  
\end{align}
This holds for all $L$, hence $\lambda_{\sigma,d} = 0$.
Since the set of infinite volume ground states is closed in the weak$^*$ topology, 
we apply equation~\eqref{eqn:gswslim} to the ground state decomposition~\eqref{eqn:gsdecomp4} of $\widehat{\omega}^{\chi,c}$,
\begin{equation}
	\begin{split}
	\widehat{\omega}^{\chi,c} = &\lambda_0 \omega^0 + \\
	&\wslim_{L\ra\infty} \bigg( \lambda_{\chi, e} \frac{\widehat{\omega}^{\chi,c}( \ \cdot \ D_L^{\chi, e})}{\widehat{\omega}^{\chi,c}(D_L^{\chi, e})} 
+c_{\iota, c}  \frac{\widehat{\omega}^{\chi,c}( \ \cdot \ D_L^{\iota, c})}{\widehat{\omega}^{\chi,c}(D_L^{\iota, c})}
+c_{\chi, c} \frac{\widehat{\omega}^{\chi,c}( \ \cdot \ D_L^{\chi, c})}{\widehat{\omega}^{\chi,c}(D_L^{\chi, c})}\bigg).
\end{split}
\end{equation}
Therefore, $\widehat{\omega}^{\chi,c} \in \conv{K^0\cup K^{\chi,e} \cup K^{\iota,c} \cup K^{\chi,c}}$.
\end{proof}

\section{Concluding remarks}\label{sec:con}

We have proved that the set of known ground states for Kitaev's quantum
double model is complete in the case of finite abelian groups.
A natural question is if the results can be extended to non-abelian groups.
The main technical challenge is that the quantum double $\mathcal{D}(G)$ has higher dimensional irreducible representations.
In physical terms, this manifests itself in the presence of \emph{non-abelian} anyons.
As a result, the structure of the ribbon operators and quasi-particles is much richer than that of the abelian case.
In particular, the quasi-particle excitations no longer decompose into simply electric and magnetic type,
since they have to account for the more complicated structure of $\operatorname{Rep}(\mathcal{D}(G))$.
For instance, the boundary operators we use in Lemmas~\ref{lem:globprojboundaryop} and~\ref{lem:gshambound}
would need to be generalized to account for this structure.
Since the fusion rules are non-abelian, the fusion of two excitations does not always have a definite outcome.
This makes it more difficult to sum over all admissible configurations that lead to a given charge in the region.
In addition, the charged sectors are generated by endomorphisms, which are not automorphisms in general, and that
are less straightforward to construct~\cite{NaaijkensLQP}.

A current challenge in mathematical physics is the classification of gapped ground state
phases \cite{BachmannMNS,BachmannO,Ogata1,Ogata2,Ogata3}.
One approach to classifying a phase is to construct a complete set of invariants.
By definition, an invariant is a quantity that is constant within a phase.
Consequently, if an invariant is computed for two systems and is found to take different values,
the systems must be in different phases. In the literature, a topological phase is often defined
as an open region in a space of Hamiltonians where there is a non-vanishing gap above the
ground state~\cite{ChenGW}. Therefore, the construction of invariants can be expected to rely on the
existence of a spectral gap. In the quantum double models, while it is known that the gap above
the vacuum state is stable under small uniform perturbations, we do not expect that the charged
ground states will survive a generic perturbation of this model since the anyon quasi-particles
will, in general, not appear as time-invariant states~\cite{Kitaev2}.
Our classification of the complete ground state space of the quantum double model for abelian groups
gives an example that shows that the set of infinite volume ground states is generally not an invariant
of a phase. From the physical point of view, however, the invariance of the structure of anyon quasi-particles
is usually taken as fact. There are few mathematically rigorous results in this direction \cite{Haah}.
We hope that our results are a first step in rigorously studying the stability properties of the superselection structure
of quantum double models.


\begin{thebibliography}{50}
\bibitem{AlickiFH}

Alicki, R., Fannes, M., Horodecki, M.: \emph{A statistical mechanics view on Kitaev's proposal of quantum memories}. 
J. Phys. A, {\bf 40}(24), 6451--6467 (2007)

\bibitem{ArakiM}

Araki, H., Matsui, T.: \emph{Ground states of the $XY$-model}. Commun. Math. Phys. {\bf 101}, 213--245 (1985)

\bibitem{ArovasSW}

Arovas, D., Schrieffer, J.R., Wilczek, F.: \emph{Fractional statistics and the quantum Hall effect}. Phys. Rev. Lett. {\bf 53}, 722--723 (1984)

\bibitem{Bachmann}

Bachmann, S.: \emph{Local disorder, topological ground state degeneracy and entanglement entropy, and discrete anyons}. arXiv:1608.03903 (2016)

\bibitem{BachmannMNS}

Bachmann, S., Michalakis, S., Nachtergaele, B., Sims, R.: \emph{Automorphic equicalence within gapped phases of quantum lattice systems}. Commun. Math. Phys. {\bf 309}, 835--871 (2012)

\bibitem{BachmannO}

Bachmann, S., Ogata, Y.: \emph{$C^1$-classification of gapped parent Hamiltonians of quantum spin chains}. Commun. Math. Phys. {\bf 338}, 1011--1042 (2015)

\bibitem{BaisDW}

Bais, F.A., van Driel, P., De Wild Propitius, M.: \emph{Anyons in discrete gauge theories with Chern-Simons terms}. Nucl. Phys. B {\bf 393}, 547--570 (1993)

\bibitem{BakalovK} 

Bakalov, B., Kirillov, Jr., A.: Lectures on Tensor Categories and Modular Functors (University Lecture Series
21). American Mathematical Society, Providence, RI (2001)

\bibitem{BeigiSW}

Beigi, S., Shor, P.W., Whalen, D.: \emph{The quantum double model with boundary: condensations and symmetries}. Commun. Math. Phys. {\bf 306}, 663--694 (2011)

\bibitem{BombinMD}

Bombin, H., Martin-Delgado, M.A.: \emph{A family of non-abelian Kitaev models on the lattice: topological condensation and confinement}. Phys. Rev. B {\bf 78}, 115421 (2008)

\bibitem{BondersonSS}
Bonderson, P., Shtengel, K., Slingerland, J.K.: \emph{Interferometry of non-abelian anyons}. Ann. Physics {\bf 323}, 2709--2755 (2008)

\bibitem{BrandaoH}
Brand{\~{a}}o, F.G.S.L, Horodecki, M.:
\emph{Exponential decay of correlations implies area law}.
Commun. Math. Phys. {\bf 333}, 761 (2015)

\bibitem{BratteliKR}

Bratteli, O., Kishimoto, A., Robinson, D.: \emph{Ground states of infinite quantum spin systems}. Commun. Math. Phys. {\bf 64}, 41--48 (1978)

\bibitem{BratteliR}

Bratteli, O., Robinson, D.W.: Operator algebras and quantum statistical mechanics 1 and 2, Second Edition. Springer Verlag (1987)

\bibitem{BravyiHM}

Bravyi, S., Hastings, M., Michalakis, S.: \emph{Topological quantum order: stability under local perturbations}. J. Math. Phys. {\bf 51}, 093512 (2011)


\bibitem{BravyiK}

Bravyi, S., Kitaev, A.: \emph{Quantum codes on a lattice with boundary}. arXiv:quant-ph/9811052v1 (1998)

\bibitem{ChenGW}

Chen, X., Gu, Z.-C., Wen, X.-G.: \emph{Local unitary transformation, long-range quantum entanglement, wave function renormalization, and topological order}, Phys. Rev. B {\bf 84}, 155138 (2011)

\bibitem{DijkgraafPR}

Dijkgraaf, R., Pasquier, V., Roche, P.: \emph{Quasi Hopf algebras, group cohomology and orbifold models}. Nucl. Phys. B. Proc. Suppl. {\bf 18B}, 60--72 (1991)

\bibitem{Haag1}

Doplicher, S., Haag, R., Roberts, J.E.: \emph{Local observables and particle statistics. I}. Commun. Math. Phys. {\bf 23}, 199--230 (1971)

\bibitem{Haag2}

Doplicher, S., Haag, R., Roberts, J.E.: \emph{Local observables and particle statistics. II}. Commun. Math. Phys. {\bf 35}, 49--85 (1974)

\bibitem{FannesNW}
Fannes, M., Nachtergaele, B., Werner, R. F.:
\emph{Finitely correlated states of quantum spin chains}.
Commun. Math. Phys. {\bf 144}, 443--490 (1992)

\bibitem{FannesW}

Fannes, M., Werner, R.F.: \emph{Boundary Conditions for Quantum Lattice Systems}. Helv. Phys. Acta {\bf 68}, 635--657 (1995)

\bibitem{FiedlerN}

Fiedler, L., Naaijkens, P.: \emph{Haag duality for Kitaev's quantum double model for abelian groups}. Rev. Math. Phys. {\bf 27}, 1550021 (2015)

\bibitem{FredRS} 

Fredenhagen, K., Rehren, K.-H., Schroer, B.: \emph{Superselection sectors with braid group statistics and exchange algebras}. Commun. Math. Phys. {\bf 125}, 201--226 (1989)

\bibitem{Freedman}

Freedman, M.: \emph{P/{NP}, and the quantum field computer}. Proc. Natl. Acad. Sci. USA {\bf 95}, 98--101 (1998)

\bibitem{FreedmanM}

Freedman, M., Meyer, D.A.: \emph{Projective plane and planar quantum codes}. Found. Comput. Math. {\bf 1}, 325--332 (2001)

\bibitem{FrohlichG}

Fr\"ohlich, J., Gabbiani, F.: \emph{Braid statistics in local quantum theory}. Rev. Math. Phys. {\bf 2}, 251--353 (1990)

\bibitem{GottsteinW}

Gottstein, C.T., Werner, R.F.: \emph{Ground states of the $q$-deformed Heisenberg ferromagnet}. arXiv:cond-mat/9501123 (1995)

\bibitem{Haag}

Haag, R.: Local quantum physics: Fields, particles, algebras, Texts and Monographs in Physics. Springer-Verlag, Berlin, second edition (1996)

\bibitem{HaagHW}

Haag, R., Hugenholtz, N.M., Winnink, M.: \emph{On the equilibrium states in quantum statistical mechanics}. Commun. Math. Phys. {\bf 5}, 215--236 (1967)

\bibitem{Haah}

Haah, J.: \emph{An invariant of topologically ordered states under local unitary transformations}. Commun. Math. Phys. {\bf 342}, 771--801 (2016)

\bibitem{Hastings}

Hastings, M.B.: \emph{An area law for one-dimensional quantum systems}.	J. Stat. Mech. {\bf 2007}, P08024 (2007)

\bibitem{HastingsK}
Hastings, M.B., Koma, T.:
\emph{Spectral gap and exponential decay of correlations}.
Commun. Math. Phys.
{\bf 265},  781--804 (2006)

\bibitem{Kitaev}

Kitaev, A.Y.: \emph{Fault-tolerant quantum computation by anyons}. Ann. Physics {\bf 303}, 2--30 (2003)

\bibitem{Kitaev2}
Kitaev, A.:
\emph{Anyons in an exactly solved model and beyond}. 
Ann. Physics {\bf 321}, 2--111 (2006)

\bibitem{KomaN}

Koma, T., Nachtergaele, B.: \emph{The complete set of ground states of the ferromagnetic XXZ chains}. Adv. Theor. Math. Phys. {\bf 2}, 533--558 (1998)


\bibitem{Matsui}

Matsui, T.: \emph{On ground states of the one-dimensional ferromagnetic XXZ chain}. Lett. Math. Phys. {\bf 37}, 397--403 (1996)

\bibitem{MichalakisP}

Michalakis, S., Zwolak, J.P: \emph{Stability of frustration-free Hamiltonians}. Commun. Math. Phys. {\bf 322}, 277--302 (2013)

\bibitem{MooreR}

Moore, G., Read, N.: \emph{Nonabelions in the fractional quantum Hall effect}. Nucl. Phys. B {\bf 360}, 362--396 (1990)

\bibitem{Naaijkens}

Naaijkens, P.: \emph{Localized endomorphisms in Kitaev's toric code on the plane}. Rev. Math. Phys. {\bf 23}, 347--373 (2011)

\bibitem{NaaijkensLQP}

Naaijkens, P.: \emph{Kitaev's quantum double model from a local quantum physics point of view}. In:  R. Brunetti C. Dappiaggi, K. Fredenhagen, J. Yngvason (eds.), Advances in Algebraic Quantum Field Theory, pp. 365--395. Springer (2015)

\bibitem{NachOS}

Nachtergaele, B., Ogata, Y., Sims, R.: \emph{Propagation of correlations in quantum lattice systems}. J. Stat. Phys. {\bf 124}, 1--13 (2006)

\bibitem{NachtergaeleS}

Nachtergaele, B., Sims, R.: \emph{Lieb-Robinson bounds and the exponential clustering theorem}. Commun. Math. Phys. {\bf 265}, 119--130 (2006)

\bibitem{Ogata1}

	Ogata, Y.: \emph{A class of asymmetric gapped Hamiltonians on quantum spin chains and its characterization I}. Commun. Math. Phys. {\bf 348}, 847--895 (2016)

\bibitem{Ogata2}

Ogata, Y.: \emph{A class of asymmetric gapped Hamiltonians on quantum spin chains and its characterization II},
Commun. Math. Phys. {\bf 348}, 897--957 (2016)

\bibitem{Ogata3}

Ogata, Y.: \emph{A class of asymmetric gapped Hamiltonians on quantum spin chains and its characterization III},
arXiv:1606.05508 (2016)

\bibitem{ReedSimon}

Reed, M., Simon, B.: Methods of Modern Mathematical Physics Vol. I: Functional Analysis, 
Revised and Enlarged edition. Academic Press, 400 p. (1980)

\bibitem{SzlachV}

Szlach\'anyi, K., Vecserny\'es, P.: \emph{Quantum symmetry and braid group statistics in $G$-spin models}. Commun. Math. Phys. {\bf 156}, 127--168 (1993)

\bibitem{Wen}

Wen, X.-G.: \emph{Vacuum degeneracy of chiral spin states in compactified space}. Phys. Rev. Lett. {\bf B40}, 7387--7390 (1989)

\bibitem{Wilczek}

Wilczek, F.: \emph{Fractional statistics and anyon superconductivity}. World Scientific Publishing Co., Inc., Teaneck, NJ (1990)

\end{thebibliography}
\end{document}